\documentclass[12pt,reqno]{amsart}
\usepackage[french, english]{babel}

\usepackage{graphicx}
\usepackage{hyperref}

\usepackage{lmodern}

\usepackage{url}
\usepackage{amssymb,amsthm,amsfonts,amstext}
\usepackage{amsmath}

\usepackage{graphicx}
\usepackage{enumerate}
\numberwithin{equation}{section}
\usepackage{a4wide}
\usepackage[active]{srcltx}
\usepackage[table]{xcolor}
\definecolor{cyan(process)}{rgb}{0.0, 0.72, 0.92}
\usepackage{xcolor}
\definecolor{columbiablue}{rgb}{0.61, 0.87, 1.0}
\usepackage{wrapfig}
\usepackage[]{pstricks}
\usepackage{fancyhdr}
\definecolor{sandstone}{HTML}{786D5F}
\definecolor{beaublue}{rgb}{0.74, 0.83, 0.9}
\definecolor{cherryblossompink}{rgb}{1.0, 0.72, 0.77}

\usepackage{color}

\usepackage[left=3.5cm,top=2.5cm,bottom=2cm,right=4cm, marginparwidth=3cm]{geometry}

\usepackage{marginnote}
\usepackage{todonotes}

\makeatletter
\newcommand*\bigcdot{\mathpalette\bigcdot@{.7}}
\newcommand*\bigcdot@[2]{\mathbin{\vcenter{\hbox{\scalebox{#2}{$\m@th#1\bullet$}}}}}
\makeatother

\usepackage{tikz}
\usetikzlibrary{backgrounds}
\usetikzlibrary{patterns,fadings}
\usetikzlibrary{arrows,decorations.pathmorphing}
\usetikzlibrary{calc}
\definecolor{light-gray}{gray}{0.95}
\usepackage{graphicx}
\usepackage{epsfig}
\usetikzlibrary{shapes}
\usetikzlibrary{decorations}

\usetikzlibrary{decorations.pathmorphing}
\usetikzlibrary{decorations.pathreplacing}
\usetikzlibrary{decorations.shapes}
\usetikzlibrary{decorations.text}
\usetikzlibrary{decorations.markings}
\usetikzlibrary{decorations.fractals}
\usetikzlibrary{decorations.footprints}

\renewcommand{\epsilon}{\varepsilon}

\newcommand\LL{{\mathbb L}}
\newcommand\EE{{\mathbb E}}
\newcommand\PP{{\mathbb P}}

\newcommand\RR{{\mathbb R}}
\newcommand\TT{{\mathbb T}}
\newcommand\ZZ{{\mathbb Z}}

\newcommand\eps{\epsilon}

\newcommand{\mc}[1]{{\mathcal #1}}
\newcommand{\mf}[1]{{\mathfrak #1}}
\newcommand{\mb}[1]{{\mathbf #1}}
\newcommand{\bb}[1]{{\mathbb #1}}

\newtheorem{proposition}{Proposition}
\newtheorem{theorem}{Theorem}
\newtheorem{lemma}{Lemma}
\newtheorem{remark}{Remark}
\newtheorem{corollary}{Corollary}
\newtheorem{definition}{Definition}

\renewcommand{\leq}{\leqslant}
\renewcommand{\geq}{\geqslant}
\renewcommand{\le}{\leqslant}
\renewcommand{\ge}{\geqslant}

\title[]{A microscopic derivation of coupled SPDE's with a KPZ flavor}

\author{R. Ahmed}
\address{Ragaa Ahmed, Center for Mathematical Analysis,  Geometry and Dynamical Systems,
Instituto Superior T\'ecnico, Universidade de Lisboa,
Av. Rovisco Pais, 1049-001 Lisboa, Portugal}
\email{\tt{ragaaahmedabbas@gmail.com}}

\author{C.Bernardin}
\address{C\'edric Bernardin, Universit\'e C\^ote d'Azur, CNRS, LJAD\\
Parc Valrose\\
06108 NICE Cedex 02, France}
\email{{\tt cbernard@unice.fr}}

\author{P.  Gon\c calves}
\address{Patr\'icia Gon\c calves, Center for Mathematical Analysis,  Geometry and Dynamical Systems,
Instituto Superior T\'ecnico, Universidade de Lisboa,
Av. Rovisco Pais, 1049-001 Lisboa, Portugal and   Institut  Henri
Poincar\'e, UMS 839 (CNRS/UPMC), 11 rue Pierre et Marie Curie, 75231 Paris Cedex 05, France.}
\email{{\tt patricia.goncalves@math.tecnico.ulisboa.pt}}

\author{M. Simon}
\address{Marielle Simon, Inria, Univ. Lille, CNRS, UMR 8524 - Laboratoire Paul Painlev\'e, F-59000 Lille.}
\email{{\tt marielle.simon@inria.fr}}

\date{\today}


\begin{document}
\maketitle

\begin{abstract}
We consider an interacting particles system composed of a Hamiltonian part and perturbed by a conservative stochastic noise so that the full system conserves two quantities: energy and volume. The Hamiltonian part is regulated by a scaling parameter vanishing in the limit. We study the form of the fluctuations of these quantities at equilibrium and derive coupled stochastic partial differential equations with a KPZ flavor.  

\end{abstract}

\tableofcontents

\section{Introduction}

During the last decade a huge number of research programs have been devoted to the study of the Kardar-Parisi-Zhang (KPZ) equation and its derivation from microscopic models. The KPZ equation has been introduced in \cite{KPZ} as a phenomenological equation which, in dimension one, takes the form 
\begin{equation*}
\partial_t h (t,u) = A\; \partial_{uu}^2 h(t,u) + B \; (\partial_u h(t,u))^2 + \sqrt{C} \; \dot{\mathcal{W}}(t,u), \qquad t >0, u \in \mathbb{R},
\end{equation*}
where $A>0,B \in \mathbb{R}, C>0$ are thermodynamic constants and $ \dot{\mathcal{W}}(t,u)$ is a standard space-time Gaussian white noise. This equation describes the evolution of a randomly growing interface, whose height is  $h(t,u)$, $t>0$ being the time and $u\in \RR$ the spatial coordinate. Almost equivalently, taking the space derivative of the KPZ equation, we get the (conservative) stochastic Burgers equation (SBE) for $\mathcal{Y} =\partial_u h$:
\begin{equation*}
\partial_t \mathcal{Y} (t,u) = A \; \partial_{uu}^2 \mathcal{Y}(t,u) + B \; \partial_u (\mathcal{Y}^2(t,u)) + \sqrt{C} \; \partial_u \dot{\mathcal{W}}(t,u).
\end{equation*}
The SBE equation is expected to be a universal object describing the scaling limit of a large class of weakly asymmetric interacting particle systems with a single conservation law \cite{Spo}. In our opinion, the major recent mathematical contributions in this field have been:
\begin{itemize}
\item the proof of the well posedness of the KPZ (or SBE) equation, via the  theory of regularity structures developed by Hairer \cite{Hai13,Hai14}, or alternatively through the paracontrolled distributions theory \cite{GPI};
\item the obtention of its asymptotic properties via the study of some  ``integrable stochastic systems", in particular the derivation of  scaling limits for one-dimensional exclusion-type processes, starting from the seminal paper \cite{BerG}, and going on with  \cite{Cor,QS}, and many others;
\item the development of a robust method to derive the SBE (or KPZ) equation as a scaling limit for a large class of interacting particle systems, thanks to the new notion of energy solutions investigated in \cite{gj2014,GubPer,GJS,GJSim}.    
\end{itemize}

Since a few years there has been a growing interest for one-dimensional $n$-component coupled SBE, written as:
\begin{equation}
\label{eq:SBE-couple}
\partial_t \vec{\mathcal{Y}}   = \partial_u ( \mathbf{A} \partial_u \vec{\mathcal{Y}}) + \partial_u ( \ll \vec{\mathcal{Y}} , \vec{\mathbf{B}} \vec{\mathcal{Y}}\gg)+ \sqrt{\mathbf{C}}\;\partial_u {\vec  \xi},
\end{equation}
where $\vec{\mathcal{Y}} (t,u) \in {\mathbb R}^n,$ and $\mathbf{A}, \mathbf{C}$  are square matrices of size $n$, $\vec \xi$ is a $n$-component Gaussian white noise, and $\vec{\mathbf{B}}=(\vec{\mathbf{B}}_{i})$ is a tensor\footnote{Therefore, in \eqref{eq:SBE-couple}, the quantity $\ll \vec{\mathcal{Y}} , \vec{\mathbf{B}} \vec{\mathcal{Y}}\gg$  is a vector whose $i$-th component reads
\[ 
 (\ll \vec{\mathcal{Y}} , \vec{\mathbf{B}} \vec{\mathcal{Y}}\gg)_i = \langle  \vec{\mathcal{Y}}, \vec{\mathbf{B}}_{i} \vec{\mathcal{Y}} \rangle,
\] where 
$\langle \cdot \rangle$ denotes the Euclidean inner product in $\RR^n$.} (\textit{i.e.}~$\vec{\mathbf{B}}_{i}$ is $n\times n$--matrix for any $1\leqslant i\leqslant n$).  Such equations appeared in the physics literature very early after the seminal paper \cite{KPZ} of Kardar, Parisi and Zhang and more recently in the context of the nonlinear fluctuating hydrodynamics theory developed by Spohn and coauthors \cite{S1,SS}. The mathematical study of global-in-time existence and invariant measures for the $n$-component coupled SBE has been investigated in \cite{GPnew,FH17,F18}.  It is expected that the coupled SBE equations cover the dynamics of weakly asymmetric interacting particle systems with several conserved quantities in a suitable mesoscopic scale. In the nonlinear fluctuating hydrodynamics theory one \textit{postulates} that these equations describe correctly the macroscopic properties of the underlying microscopic system in ``some mesoscopic scale" and their study permits to obtain some information of the large time behavior of the \textit{strongly} asymmetric system. Let us notice that it is quite unclear (at least for us), even without asking for a proof, what are the exact scaling limits to perform in the microscopic models in order to obtain these equations. Indeed, such equations should be obtained by tuning in a very specific way the intensity of the ``asymmetry" and the time scale with respect to the scaling parameter. Moreover, since the system has more than one conserved quantities, different time scales have to be considered. Let us also remark that the mathematical treatment of these equations and their obtention as scaling limits are challenging problems whose  resolutions are in their very infancy.

The aim of the paper is to provide a model with two conserved quantities for which it can be proved rigorously that in suitable scaling limits, the system is described by a set of degenerate coupled SBE  equations. By degenerate we mean that some of the matrices entries appearing in \eqref{eq:SBE-couple} vanish. When the asymmetry is very weak, in a diffusive time scale, the system is reduced to an uncoupled system consisting of two autonomous Ornstein-Uhlenbeck (OU) equation. If the intensity of the asymmetry is increased, there is some  critical value such that, in a diffusive time scale, the system becomes composed of coupled equations: an autonomous OU equation and a second OU equation with a drift term driven by the first one. Increasing again the intensity of the asymmetry, the system becomes composed of an OU equation (obtained in a diffusive time scaling) and a transport equation whose transport term is driven by the first one (obtained in a sudiffusive time scaling). This pictures remains valid up to a second critical value  of the asymmetry intensity, for which the first OU equation is replaced by a SBE equation while the second is still a transport equation driven by some OU process. The results obtained are in agreement with mode coupling theory \cite{Gunter}.        

This paper is one of the first contributions where coupled equations with a KPZ flavor are derived from a microscopic system. In \cite{BFS}, the authors derive some multicomponent coupled SBE equations as a scaling limit of a multi-species zero-range process. However a big difference of our model with respect to the latter is that in \cite{BFS} the velocities of the normal modes are equal while it is never the case in our model. A second interesting feature of our result is that we are able to emphasis the exact time and asymmetry parameters scaling to consider in order to get the expected equations, and we  extend the class of SPDEs which arise in this context.

\subsection*{Outline of the paper} We start in Section \ref{sec:model} with the definition of the microscopic dynamics under investigation and the introduction of the relevant macroscopic quantities. Section \ref{sec:spdes} is devoted to defining and giving a rigorous meaning to the solution of three stochastic partial differential equations which will emerge at the macroscopic level. In Section \ref{sec:stat_res} we will state our main convergence results. Finally, Sections \ref{sec:sketch}, \ref{sec:limit} and \ref{sec:tight} contain the different steps of the proof: we begin with a sketch given in Section \ref{sec:sketch}, and we are able conclude the proof up to technical results, namely the convergence of martingales associated to the microscopic dynamics (proved in Section \ref{sec:limit}), and the tightness property of the fluctuation fields (proved in Section \ref{sec:tight}).

\subsection*{Notations}

Given two real-valued functions $f$ and $g$ depending on the variable $u \in \bb R^d$ we will write $f(u) \approx g(u)$ if there exists a constant $C>0$ which does not depend on $u$ such that for any $u$, $C^{-1} f(u) \le g(u) \le C f(u)$ and $ f(u) \lesssim g(u)$ if for any $u$, $f(u) \le C g(u)$. We write $f =\mc O (g)$ (resp. $f=o(g)$) in the neighborhood of $u_0$ if $| f| \lesssim | g|$ in the neighborhood of $u_0$ (resp. $\lim_{u \to u_0} f(u)/g(u) =0$). Sometimes it will be convenient to precise the dependence of the constant $C$ on some extra parameters and this will be done by the standard notation $C(\lambda)$ if $\lambda$ is the extra parameter. We often denote the one-dimensional  gradient and Laplacian on $\mathbb R$ by $\nabla=\partial_u$ and $\Delta=\partial_{uu}^2$. The transpose matrix of the matrix $\mathbf{A}$ is denoted by $\mathbf{A}^\dagger$.   The one-dimensional continuous torus is denoted by $\TT=[0,1)$. 
For any integer $d\ge 1$, we denote the space of smooth $\RR^d$-valued functions $\vec{ f}:=(f^1, \ldots, f^d)^\dagger$ on $\TT$ by  $\mc D (\TT, \RR^d)$. 
 In the special case $d=1$, we simplify the notation by omitting the arrow on $f$, and we  also denote $\mc D (\TT, \RR)$ (resp. $\bb L^2(\bb T, \RR)$) by $\mc D (\TT)$ (resp. $\bb L^2(\bb T)$). Then, we identify $\mc D (\bb T, \RR^d)$ with $(\mc D (\bb T))^d$ and $\LL^2(\TT,\RR^d)$ with $(\LL^2(\TT))^d$. Finally, for a function $\vec f \in \LL^2(\bb T, \RR^d) $, we  denote by $\| \vec f\|_0^2$ the usual $\bb L^2(\bb T, \RR^d)$-norm of $\vec f$: \[\|\vec f\|_0^2:=\sum_{i=1}^d \int_\bb T (f^i (u))^2\, du,\] and by $\langle \cdot , \cdot \rangle_0$ its associated inner product: 
 \[ 
 \langle \vec f , \vec g \; \rangle_0 = \sum_{i=1}^d \int_{\TT} f^i(u) g^i(u) du.
 \]

\section{The model}
\label{sec:model}

Let $b>0$ be a fixed parameter and define the one-dimensional exponential potential 
$$V_{b}: u \in \RR \to e^{-bu} -1+bu \in [0, +\infty).$$
Recall that the one-dimensional continuous torus is denoted by $\TT=[0,1)$. For any $n\ge 1$, we define its discrete counterpart $\TT_n =\{0,1, \ldots, n-1\}$ of size $n$ and we denote $\RR^{\TT_n}$ by $\Omega_n$. We consider the Markov process $\eta(t)=\{\eta_x(t):x\in {\mathbb T}_n\}$ with state space $\Omega_n$ defined by its infinitesimal generator $\mc L$. The latter is given by 
$$\mc L=\alpha_n \mc A+\gamma \mc S,$$ 
where $\gamma>0$ and $\alpha_n=\alpha n^{-\kappa}$, with $\alpha\in\mathbb R$, $\kappa>0$.  The actions of $\mc A$ and $\mc S$ on differentiable functions $f:\Omega_n \rightarrow{\mathbb{R}}$ are given by
\begin{equation*}
(\mc Af)(\eta)=\sum_{x \in \TT_n} \big(V_b^{\prime} (\eta_{x+1}) -V_b^{\prime} (\eta_{x-1})  \big) (\partial_{\eta_x} f)(\eta)
\end{equation*}
and
\begin{equation*}
 (\mc Sf)(\eta)=\sum_{x \in \TT_n} \big( f(\eta^{x,x+1}) -f(\eta) \big).
\end{equation*}
Here the configuration $\eta^{x,x+1}$ is the configuration obtained from $\eta$ by exchanging the occupation variables $\eta_x$ and $\eta_{x+1}$, i.e. for any $z \in \TT_n$, $(\eta^{x,x+1})_z =\eta_z$ for $z\ne x, x+1$, $(\eta^{x,x+1})_x = \eta_{x+1}$ and $(\eta^{x,x+1})_{x+1} = \eta_{x}$. We refer the interested reader to \cite{B0,BS,BG,SS} for the motivations behind the study of this system and more information about its construction. The system is thus a Hamiltonian system (with generator $\mc A$) perturbed by a stochastic noise (generated by $\mathcal S$). 

\medskip

Let us comment about the role of the parameters which appear in the definition of the microscopic system. In the following, $n$ will tend to infinity so that $1/n$, which represents the ratio between the macroscopic scale and the microscopic scale, will play the role of a scaling parameter going to $0$. The parameter $\alpha_n$ fixes the intensity of the asymmetry in the system in terms of the scaling parameter: larger $\kappa$ is, smaller the asymmetry is. Strong asymmetric systems would correspond to $\kappa=0$. The reason why $\mathcal A$ represents the asymmetric part of the generator will become clear in the sequel. The parameter $\gamma$ fixes the intensity of the stochastic noise and is always of order $1$. We will consider the Markov process in different time scales, namely by accelerating the microscopic time by a constant $\theta (n)=n^a$, where $a>0$ is a constant. 

\medskip

The system conserves two quantities: the \emph{energy} and the \emph{volume}, given respectively by $$\sum_{x \in \TT_n} V_{b} (\eta_x),\quad \quad \sum_{x \in \TT_n} \eta_x.$$
The previous conservation laws are expressed by the well defined continuity equations 
\begin{equation}
{\mc L} (V_{b} (\eta_x))= {\bar j}^e_{x-1,x}(\eta) - {\bar j}^e_{x,x+1}(\eta), \qquad \mc L (\eta_x)={\bar j}^{v}_{x-1,x}(\eta) - {\bar j}^{v}_{x,x+1}(\eta),
\end{equation}
where the microscopic currents are given by
\begin{align}
{\bar j}^e_{x,x+1}(\eta) & =-\alpha_n b^2 e^{-b(\eta_x + \eta_{x+1})}+\alpha_n b^2(e^{-b \eta_x} +e^{-b \eta_{x+1}})- \gamma \nabla (V_{b} (\eta_{x})) \\
{\bar j}^{v}_{x,x+1}(\eta) & =\alpha_n b (e^{-b \eta_x} + e^{-b\eta_{x+1}})-\gamma \nabla \eta_x .
\end{align}
We define a family of product probability measures $\mu_{\bar \beta, \bar \lambda}$ {{on $\Omega_n$}} by
\begin{equation}\label{eq:inv_measure}
\mu_{{\bar \beta}, {\bar \lambda}} (d\eta) = \prod_{x\in \bb T_n} {\bar Z}^{-1} ({\bar \beta},{\bar \lambda}) \exp \{ -{\bar \beta} e^{-b \eta_x} - {\bar\lambda} \eta_x\} d\eta_x, \quad {\bar \beta}, {\bar\lambda}>0,
\end{equation}
where $\bar Z(\bar \beta,\bar \lambda)$ is the normalization constant. 
It is a simple exercise to show that $\mathcal A$ is skew symmetric and $\mathcal S$ is symmetric in $\mathbb L^2 (\mu_{\bar \beta, \bar \lambda})$ so that $\mu_{{\bar \beta}, {\bar \lambda}}$ is an invariant measure for the dynamics generated by $\mathcal L$. In fact $\mathcal A$ is a Liouville operator corresponding to a Hamiltonian dynamics with Gibbs measure $\mu_{{\bar \beta}, {\bar \lambda}}$ and $\mathcal S$ generates the dynamics of a reversible Markov process with respect to $\mu_{{\bar \beta}, {\bar \lambda}}$. Therefore, $\mathcal A$ represents the asymmetric part of the system.

\medskip

Let $\langle \cdot\rangle$ denote the average with respect to $\mu_{\bar \beta, \bar \lambda}$. Let us introduce the quantity \[\xi_x=e^{-b\eta_x}.\] Note that if $\eta$ is distributed according to \eqref{eq:inv_measure} then $\xi$ is distributed according to the probability measure  $\nu_{\beta,\lambda}$
given by 
\begin{equation}\label{eq:inv_mea_xi}
\nu_{\beta,\lambda}(d\xi)=\prod_{x\in\bb Z}Z^{-1}(\beta,\lambda)\textbf{1}_{\{\xi_x>0\}}e^{-\beta \xi_x+\lambda\log(\xi_x)}d\xi_x
\end{equation}
with  $\beta=\bar{\beta}$ and $\lambda=-1+\bar{\lambda}/b$. Above $Z(\beta,\lambda)$ is a normalizing constant.

\medskip
We are interested in the evolution of this  process in some accelerated time scale $t\theta(n)$, thus we denote by $\{\eta(t\theta(n))\,; \, t\in [0,T]\}$ the Markov process on $\Omega_n$ associated to the accelerated generator $\theta(n) \mathcal{L}$. The path space of c\`adl\`ag trajectories with values in $\Omega_n$ is denoted by $\mathbb{D}([0,T],\Omega_n)$. We denote by $\bb P$ the probability measure on  $\mathbb{D}([0,T],\Omega_n)$ induced by an equilibrium initial condition $\mu_{\bar\beta,\bar\lambda}$ and the Markov process $\{\eta(t\theta(n))\; ;\; t \in [0,T]\}$. The corresponding expectation is denoted by $\bb E$. 

\medskip

We define ${e}:={e} ({\bar \beta}, {\bar \lambda})$ and ${v}:= {v} ({\bar \beta}, {\bar \lambda})$ as the averages of the conserved quantities $V_b (\eta_x)$, $\eta_x$ with respect to $\mu_{\bar \beta, \bar \lambda}$, respectively, namely \[{e}=\langle V_b(\eta_x) \rangle, \qquad {v}=\langle \eta_x \rangle.\] We also define $\rho$ as the average of $\xi_x$ with respect to $\mu_{\bar\beta,\bar\lambda}$, \textit{i.e.}~$\rho=\langle \xi_x\rangle$.  Finally, we denote the variance of $\eta_x$ (resp.~$\xi_x$) with respect to $\mu_{\bar \beta, \bar \lambda}$ by $\sigma^2$ (resp.~$\tau^2$) and the covariance between $\eta_x$ and $\xi_x$ by $\delta$.
 To summarize, we have   
\begin{align}
&\langle \eta_x \rangle =v \; \quad \qquad \qquad \qquad \qquad \qquad \;\textrm{and}\quad \langle ( \eta_x-v)^2\rangle=:\sigma^2\\
&\langle \xi_x\rangle =1+e-b v=\frac{\lambda+1}{\beta}=:\rho\quad\quad \textrm{and}\quad  \langle ( \xi_x-\rho)^2\rangle=\frac{\lambda+1}{\beta^2}=:\tau^2\\
&\langle (\eta_x -v) (\xi_x -\rho) \rangle =: \delta.
\end{align}
A simple computation shows that
 \begin{align}
 \langle {\bar j}_{x,x+1}^e \rangle & =-\alpha_n b^2({e} -b {v} )^2 + \alpha_n b^2 \\ \langle {\bar j}^{v}_{x,x+1}  \rangle & =2\alpha_n b (1+{e} -b {v}).
\end{align}
Hence, in the hyperbolic scaling $\theta(n)=n$, in the strong asymmetry regime, namely $\kappa=0$, the hydrodynamical equations  are given by (see \cite{BS} for a proof):
\begin{equation}
\begin{cases}
\partial_t {\mb e} - \alpha b^2 \, \partial_u ( ({\mb e} - b{\mb v})^2) =0\\
\partial_{t} {\mb v} + 2 \alpha b \, \partial_u ({\mb e -b {\mb v}}) =0.
\end{cases}
\end{equation}

\section{Stochastic partial differential equations}
\label{sec:spdes}

In this section we give the rigorous meaning of the various SPDEs which will appear in the scaling limits of our system. Let us start with a few notations.

\medskip

 The topological dual of a topological space $E$ is denoted by $E^\prime$. Hence, the space of $\RR$-valued distributions on $\TT$ is denoted by $\mc D' (\TT)$. Similarly, the space of $\RR^d$-valued distributions on $\TT$ is denoted by $(\mathcal{D}(\TT,\RR^d))'$.   If $f\in\mc D(\bb T)$ and $\mc Z=(\mc Z^1,\cdots, \mc Z^d)^\dagger \in (\mc D'(\bb T))^d$, then we denote by $\mc Z(f)$ the vector $(\mc Z^1(f),\cdots, \mc Z^d(f))\in \bb R^d$.
 
\begin{definition} For any $\mc Z=(\mathcal{Z}^1,\dots, \mc Z^d)^\dagger \in (\mc D'(\TT))^d$, we define the element ``$\mc Z \bigcdot$" belonging to $(\mathcal{D}(\TT,\RR^d))'$ by 
 \begin{equation}
 \label{eq:Zdot}
  \mc Z \bigcdot\vec  f = \sum_{j=1}^d \mc Z^j (f^j),  \quad \text{for any  } \vec f = (f^1,\dots,f^d)^\dagger \in \mathcal{D}(\TT,\RR^d).
 \end{equation}
 \end{definition}
 
Since we are going to consider time processes, let us now define the space $\mathbb{D}([0,T], (\mathcal{D}'(\bb T))^d)$ (resp.  $\mathcal{C}([0,T], (\mathcal{D}'(\bb T))^d $) as the space of $(\mathcal{D}'(\bb T)^d)$-valued functions with c\`adl\`ag (resp. continuous) trajectories.   We equip these spaces with the uniform weak topology: a sequence $\{\mc Z_\cdot^n\}_{n\ge 1}$ converges to a path $\mc Z_\cdot$ if for all $f \in \mc D (\bb T)$, we have
\begin{equation*}
\lim_{n \to \infty} \; \sup_{0\le t \le T} \; \Big| \mc Z^n_t ( f) - \mc Z_t ( f) \Big|=0,
\end{equation*}
where $|.|$ denotes the usual euclidean norm  in  $\bb R^d$.
We define similarly  the space $\mathbb {D}([0,T],(\mathcal{D}(\bb T, \RR^d))^\prime)$ (resp.~$\mc C ( [0,T], (\mathcal D (\bb T, \RR^d))^\prime)$) as the space of $(\mathcal{D}(\bb T, \RR^d))^\prime$-valued functions with c\`adl\`ag (resp.~continuous) trajectories and we endow them with
 the uniform weak topology.  For any $\mc Z \in \bb D ( [0,T], (\mathcal{D}' (\bb T))^d )$ we define the element ``$\mc Z_\cdot \bigcdot$"  belonging to the space $\mathbb{D}([0,T], (\mathcal{D}(\bb T, \RR^d))^\prime )$ by the same definition as in \eqref{eq:Zdot}. 
 
 Note that  a sequence $\{\mc Z^n\}_n$ taking values in $\bb D ( [0,T], (\mathcal{D} '(\bb T))^d )$ converges to some $\mc Z$ if and only if the sequence $\{\mc Z^n  \bigcdot\, \}_n$ of $\bb D ( [0,T], (\mathcal{D} (\bb T,\bb R^d))' )$ converges to the element $\mc Z \bigcdot $.

{We recall the following standard definition: 
\begin{definition} Let $\{\mc Z_t \in (\mathcal{D}' (\bb T))^d \; ; \; t \in [0,T]\}$ be a process.  We say that $\mc Z$ is a \emph{centered Gaussian process} if any linear combination of the components of  $\{ \mc Z_{t_i} (f_i)\, ;\, i=1, \ldots,k\}$, with $f_i \in \mc D(\bb T)$, is a Gaussian random variable.
\end{definition}}

\begin{definition}
We say that the stochastic process 
$$\{\mc B_t= (\mc B_t^1, \ldots, \mc B_t^d)^\dagger \; ; \; t\in [0,T]\}$$
whose paths are in $\mathcal{C} ( [0,T], (\mathcal{D}' (\bb T))^d )$ is a \emph{standard  $(\mathcal{D}' (\bb T))^d$-valued Brownian motion} if it is a centered Gaussian process such that 
\begin{equation*}
\forall (s,t) \in [0,T]^2,  \quad \mathbb{E}\big[ \mc B_t (f) \mc B_s^\dagger (g)\big] = (s\wedge t)\;  {\rm {I_d}} \; \int_{\mathbb{T}} f(u) g(u) du
\end{equation*}
where ${\rm{I_d}}$ is the identity matrix of size $d$ and any $f,g\in\mc D(\bb T).$
\end{definition}

\subsection{The Ornstein-Uhlenbeck equation}
Let $\mc B$ be a standard $(\mathcal{D}' (\bb T))^d$-valued Brownian motion. The first SPDE which we would like to make sense of is the $d$-dimensional Ornstein-Uhlenbeck equation, formally written as:
\begin{equation}
\label{eq:ou_ini}
d\mc Z_t=\mf A\, \Delta \mc Z_t dt+ \sqrt {2\mf C}\, \nabla d\mc B_t,
\end{equation}
where $\mf A, \mf C$ are symmetric non-negative $d$-squared matrices. 

\begin{definition}
\label{def:uni_OU}
We say that the stochastic process $\{\mc Z_t \; ; \; t \in [0,T]\}$ taking values in the space $\mc C( [0,T], (\mathcal{D} '(\bb T))^d)$ is a \emph{stationary solution} of \eqref{eq:ou_ini} if it satisfies:
\begin{enumerate}[i)]
\item  For every $\{\vec f_t :  \bb T\rightarrow \bb R^d \; ; \; t \in [0,T]\}$ which is $C^1$ in time and smooth in space,  the quantity given by
\begin{equation}\label{lf1}
\mc M_t   \bigcdot \vec f_\cdot= \mc Z_t \bigcdot \vec f_t -\mc Z_0 \bigcdot \vec f_0 - \int_0^t \mc Z_s \bigcdot ({\mf A}^\dagger \Delta \vec f_s) \, ds\, ,
\end{equation}
is a martingale with respect to the natural filtration associated to $\mc Z_\cdot$, namely 
\begin{equation}
\mc F_t:=\sigma \big( \mc Z_s \bigcdot \vec f \; ; \; s \leqslant t, \vec f \in \mc D(\TT, \RR^d)\big), \label{eq:filtration}
\end{equation}
with quadratic variation equal to \[2\int_0^t \,\big\|\sqrt \mf C\,  \nabla \vec f_s \big\|_{0}^2\,ds.\]
\item $\mc Z_0$ is a mean zero Gaussian field such that for any $\vec f, \vec g\in {\mathcal{D}(\mathbb{T}, \RR^2)}$, by 
\begin{equation}
\label{eq:covar1}
\mathbb{E}\big[ (\mc Z_0 \bigcdot \vec f) \;  (\mc Z_0 \bigcdot \vec g) \big] = \big\langle \vec f \, , \, {\mf D} \vec g \big\rangle_0
\end{equation}
where $\mf D$ is the (symmetric matrix) solution of
\begin{equation*}
\mf A \mf D+\mf D\mf A^\dagger=2\mf C.
\end{equation*}
\end{enumerate}
\end{definition}

\begin{remark}
A simple computation shows that, if $\mc Z_\cdot$ satisfies Definition \ref{def:uni_OU}, then  $\mc Z\bigcdot$ is a centered Gaussian process with covariance given by
\begin{equation}
\label{eq:covar2}
\mathbb{E}\big[ (\mc Z_t \bigcdot \vec f) \;  (\mc Z_s \bigcdot \vec g) \big] = \big\langle T_{t-s}\vec f \, , \, {\mf D} \vec g \big\rangle_0
\end{equation}
where $T_t:=\exp(t\; \mf A^\dagger\Delta)$.
\end{remark}

\begin{proposition}
\label{prop:uni_OU}
There exists a unique stationary solution to \eqref{eq:ou_ini} in the sense of Definition \ref{def:uni_OU}. It is called a stationary $d$-dimensional generalized Ornstein-Uhlenbeck process. 
\end{proposition}

\subsection{A two-dimensional drifted  Ornstein-Uhlenbeck equation}

Let $c\in \RR$ be fixed. We define the shift operator $T^\pm_c$ acting on functions $\vec f\in \mc D(\bb T, \RR^d)$ by
\begin{equation*}
\big( \, T_c^\pm \vec f \, \big) (x) =\vec f ( x \pm c ), \quad x \in \TT.
\end{equation*}
Let $\mc B$ be a standard $(\mathcal{D} '(\bb T))^2$-valued  Brownian motion. We consider positive real numbers $\lambda, \mu, \mf a, \mf d>0$ and $\mf b, \theta \in \RR$ such that $\mf a \mf d -\mf b^2 >0$. For each time $t>0$ we define the time dependent operators acting on functions $\vec f : \TT \to \RR^2$ by 
\begin{equation}\label{eq:operator_L_t}
{\mf L}_t : \vec f \mapsto 
\left( \begin{array}{cc}
 \lambda \Delta & 0 \\ \theta \nabla T_{ct}^+ &\mu \Delta \end{array}\right) \; \vec f
\end{equation}
and
\begin{equation}\label{eq:C_t}
{\mf C}_t : \vec f \mapsto 
\left( 
\begin{array}{cc}
\mf a  & \mf b T^-_{ct}\\
\mf b  T^+_{ct} & \mf d 
\end{array}
\right) \; \vec f.
\end{equation}
Observe that $\mf C_t$ is a non-negative symmetric operator: for any $\vec f =(f^1, f^2)^\dagger \in \bb L^2 (\TT, \RR^2)$, we have that
\begin{align}
\mf q_t (\vec f)&:= \big\langle \; \vec f, \; \mf C_t \vec f \; \big\rangle_0 \vphantom{\int}  \label{eq:defq}\\
&= \mf a  \int_\TT \big(f^1 (y) \big)^2 dy + \mf d   \int_\TT \big(f^2 (y) \big)^2\ dy+ 2 \mf b  \int_\TT \big(T_{ct}^+ f^1 (y) \, f^2 (y) \big) dy  \notag \\
& \ge 0 \vphantom{\int} \notag
\end{align}
because $\mf a, \mf d>0$ and $\mf a \mf d - \mf b^2 >0$.
The adjoint operator of $\mf L_t$ in $\bb L^2 (\TT, \RR^2)$ is denoted by $\mf L_t^\dagger$ and is given by 
\begin{equation*}
{\mf L}_t^\dagger : \vec f \mapsto 
\left( \begin{array}{cc}
 \lambda \Delta & -\theta \nabla T_{ct}^- \\ 0 &\mu \Delta \end{array}\right) \; \vec f.
\end{equation*}
In this section we want to make sense of the two-dimensional coupled SPDE system
\begin{equation}
\label{eq:2DOU}
d\mc Z_t= {\mf L}_t \mc Z_t dt + \sqrt {2{\mf C}_t}\, \nabla d\mc B_t.
\end{equation}


\begin{definition}
\label{def:ed}
We say that the stochastic process $\{\mc Z_t \; ; \; t \in [0,T]\}$ taking values in the space $\mc C([0,T],(\mathcal{D}' (\bb T))^2)$ is a solution of \eqref{eq:2DOU} with initial condition ${\mc Z}_0$ if for every function $\vec f:[0,T]\times \bb T\rightarrow \bb R^d$ which is  $C^1$ in time and smooth in space, the quantity given by
\begin{equation}
\label{eq:mar_ed}
\mc M_t \bigcdot \vec f_\cdot= \mc Z_t\bigcdot \vec f_t - \mc Z_0\bigcdot \vec f_0 -\int_0^t  \mc Z_s\bigcdot (\partial_s \vec f_s+\mf L_s^\dagger \vec f_s) \,ds\,,
\end{equation}
is a martingale with respect to the natural filtration associated to $\mc Z_\cdot$ as in \eqref{eq:filtration}, with quadratic variation \[\int_0^t \mf q_s( \nabla \vec f_s)ds,\] where $\mf q$ has been defined in \eqref{eq:defq}.    
\end{definition}

\begin{remark} \label{rem:trivial}
We observe that, when  the equation \eqref{eq:2DOU} has no noise ($\mf a = \mf b=\mf d=0$) and no diffusive part ($\lambda=\mu=0$) that is:
\begin{equation}
\label{eq:trivial_0}
\mc Z^1_t(f) -\mc Z^1_0(f) =0, \qquad
\mc Z^2_t(f) -\mc Z^2_0(f) =\theta  \int_0^t \mc Z^1_0(\nabla T^-_{c s}f)\,ds,
\end{equation}
then Definition \ref{def:ed} remains in force, the only difference being in \eqref{eq:mar_ed}, where the martingale term now is not present. We call this equation the \emph{trivial transport equation with parameter $\theta$}.
\end{remark}

\begin{proposition}
\label{prop:ed}
There exists a unique stochastic process $\mc Z$ solution of \eqref{eq:2DOU} in the sense of Definition \ref{def:ed}. Moreover, if $\mc Z_0$ is a Gaussian field then $\mc Z\bigcdot$ is a Gaussian process.
\end{proposition}
\begin{proof}
We fix $t\in[0,T]$  and we define the semigroup $P^{(t)}_s$ by $$\partial_sP^{(t)}_s \vec f=\mf L^\dagger_{t-s}P^{(t)}_s\vec f$$ where $\vec f\in{\mc D(\bb T, \bb R^2)}$.  Now we fix $\vec f\in\mc D(\bb T,\bb R^2)$ and we  apply \eqref{eq:mar_ed} to the time dependent function $\vec f^{(t)}_s = P^{(t)}_{t-s}\vec f$, which satisfies $\partial_s \vec f^{(t)}_s = -\mf L_s^\dagger \vec f^{(t)}_s$ and $ \vec f_0^{(t)}=P_t^{(t)} \vec f$, $\vec f_t^{(t)}=\vec f$. We get that
\begin{equation}
\Big\{ \mc M_s \bigcdot \vec{f}_\cdot\Big\}_{0\le s \le t} = \Big\{ \mc Z_s\bigcdot \vec f^{(t)}_s - \mc Z_0\bigcdot P^{(t)}_t\vec f\Big\}_{0\le s \le t}
\end{equation}
is a martingale with a deterministic quadratic variation given by 
\[\int_0^s \mf q_u \big( \nabla P^{(t)}_{t-u}\, \vec f \big) \, du.\]
From this we conclude that for any $s \le t$, 
\begin{align*}
\bb E\Big[\exp\{i\mc Z_t\bigcdot \vec f_t^{(t)} \}\Big| \mc F_s\Big]&=\exp\{i \mc Z_0\bigcdot P^{(t)}_t \vec f\}\bb E\Big[\exp\{i\mc M_t\bigcdot \vec f_{\bigcdot} \}\Big| \mc F_s\Big]\\&=
\exp\{i \mc Z_0\bigcdot P^{(t)}_t \vec f \}\exp\Big\{-\frac 12 \int_s^t \mf q_v (\nabla P^{(t)}_{t-v}\vec f) dv\Big\}\\
&\quad \times \exp\{i \mc M_s\bigcdot \vec f_\cdot \}\\
&=\exp\{i \mc Z_s \bigcdot \vec f_s^{(t)} \}  \exp\Big\{-\frac 12 \int_s^t \mf q_v (\nabla P^{(t)}_{t-v}\vec f) dv\Big\}.
\end{align*}
In the second equality above we used the fact that 
$$ \left\{ \exp\Big\{i\mc M_s \bigcdot \vec f_\cdot +\frac 12\int_0^s \mf q_v (\nabla P^{(t)}_{t-v}\vec f)\, dv\Big\}\; ; \; s\in [0,t]\right\}$$ 
is a martingale with respect to $(\mathcal F_s)_s$. 
By tower property of the conditional expectation we can then deduce an explicit expression for the finite dimensional distributions of the process $\mc Z \bigcdot$ whose only free parameter is the initial distribution of $\mc Z_0$ and this shows uniqueness. Moreover if $\mc Z_0$ is a Gaussian field, the characteristic function of the finite dimensional distributions of $\mc Z\bigcdot$ takes the form of the characteristic function of a Gaussian random vector. This completes the proof of the proposition.
\end{proof}

%
%
%

%

\subsection{The one-dimensional stochastic Burgers equation}
Now we want  to make sense of the 1-dimensional stochastic Burgers equation (SBE):
\begin{equation}\label{eq:sbe_ini}
d\mc Y_t=A\, \Delta \mc Y_t dt+B\, \nabla (\mc Y_t^2) dt + \sqrt {2C}\, \nabla d\mc B_t
\end{equation}
where $A>0,B \in \RR$ and $C>0$ are constants, and $\mc B$ is a  $\mc D'(\TT)$-valued standard Brownian motion.   
 
\begin{definition}
\label{def:sbe}
A stochastic process $\{\mathcal{Y}_t\,;\, t\in[0,T]\}$ taking values in $\mathcal{C}([0,T],\mathcal{D}'(\mathbb{T}))$ is a \emph{stationary energy solution} of \eqref{eq:sbe_ini} if
\begin{enumerate}[i)]
\item for each $t\in [0,T]$, $\mathcal{Y}_t$ is a $\mc D'(\bb T)$-valued white noise with variance $C/A$;
\item there exists a constant $\kappa>0$ such that for any $f\in \mc D(\bb T)$  and $0<\delta<\epsilon<1$
\begin{equation}\label{eq:energy_estimate}
\mathbb{E}\Big[\big(\mathcal{Q}_{s,t}^\epsilon(f)-\mathcal{Q}_{s,t}^\delta(f)\big)^2\Big]\leq \kappa \epsilon  (t-s)\|\nabla f\|_{0}^2, \end{equation}
where $$\mathcal{Q}_{s,t}^\epsilon(H):=\int_s^t\int_\mathbb{T}\big(\mathcal{Y}_r (\iota_\epsilon(u)) \big)^2\;\nabla f(u)\,du\,dr$$
and for $u\in\bb T$ the function $\iota_\varepsilon(u): [0,1]\to \bb R$ is the  approximation of the identity defined as 
\[
\iota_\varepsilon(u)(v):= \varepsilon^{-1}  \; \mathbf{1}_{]u,u+\varepsilon]}(v)  
\] 
\item for $f\in\mc D(\mathbb{T})$,  \[\mathcal{Y}_t(f)-\mathcal{Y}_0(f)-A\int_0^t  \mathcal{Y}_s  (\Delta f) ds + B\mathcal  Q_t(f)\]
is a Brownian motion of variance $2C t\|\nabla f\|_{0}^2$,  where $\mathcal{Q}_t(f):=\lim_{\epsilon\to 0}\mathcal{Q}_{0,t}^\epsilon(f)$, the limit being in $\mathbb{L}^2$;
\item the reversed process $\{\mathcal{Y}_{T-t}\,;\, t\in[0,T]\}$ satisfies item $iii)$ with $B$ replaced by $-B$.
\end{enumerate} 
\end{definition}

\begin{proposition}[Theorem 2.4, \cite{GubPer}]
\label{prop:sbe}
There exists a unique random element $\mc Y$ which is a stationary energy solution of \eqref{eq:sbe_ini} in the sense of Definition \ref{def:sbe}.
\end{proposition}

\section{Statement of results}\label{sec:stat_res}
\subsection{Topological setting}
For each integer $z\in\ZZ$, let \begin{equation} \label{eq:defh}h_{z}: x \in \TT \mapsto \begin{cases} \sqrt 2 \cos(2\pi z x)& \text{ if } z >0, \\ \sqrt 2 \sin(2\pi zx) & \text{ if } z < 0, \\
1 & \text{ if } z=0. \end{cases} \end{equation}  The set $\{h_{z}\; ; \;  z\in \bb Z\}$ is an orthonormal
basis of $\bb L^{2}(\mathbb{T})$. Consider in $\bb L^{2}(\mathbb{T})$ the operator $\mc K=(\mathrm{Id}-\Delta)$. A simple computation shows that
$\mc Kh_{z}=\gamma_{z}h_{z}$ where $\gamma_{z}=1+4\pi^2|z|^2$.

For any integer $k\geq{0}$, denote by $\bb {H}_{k} \subset \bb L^2 (\TT)$ the Hilbert space induced by $\mc D (\mathbb{T})$ and the scalar product $\langle \cdot,\cdot \rangle_{k}$ defined by $\langle f,g\rangle_{k}=\langle f,\mc K^{k}g\rangle_0$,  which can be written  as 
\begin{equation}
\label{eq:inner_product_H_k}
\langle f,g\rangle_{k}=\sum_{z\in\bb Z}\langle f,h_z\rangle_0\,  \langle g, h_z\rangle_0\;  \gamma_z^k, \qquad f,g \in \mc D(\TT).
\end{equation}
Denote by $\bb {H}_{-k} \subset {\mc D}^\prime (\TT)$ the dual of $\bb{H}_{k}$ relatively to the scalar product $\langle \cdot, \cdot \rangle_0$. It is a Hilbert space for the inner product $\langle \cdot, \cdot \rangle_{-k}$ defined by
\begin{equation}
\label{eq:inner_product_H_minus_k}
\langle \mc Y^1,\mc Y^2\rangle_{-k}=\sum_{z\in\bb Z}  \mc Y^1 (h_z) \,  \mc Y^2 (h_z ) \; \gamma_z^{-k}, \qquad \mc Y^1, \mc Y^2\in \bb H_{-k}.
\end{equation}
We denote by $\|\cdot\|_{-k}$ the corresponding norm.
We generalize the previous spaces when $\bb L^2 (\TT)$ is replaced by $(\bb L^2 (\TT))^2$.
 The set $\{ (h_z,0)^\dagger, (0,h_z)^\dagger \; ;\; z \in \bb Z\}$ is then an orthonormal basis of $(\bb L^{2}(\mathbb{T}))^2$.
We  define the space $\bb H_k\times \bb H_k$ as the Hilbert space induced by $(\mc D (\mathbb{T}))^2$  and the scalar product $\langle \cdot,\cdot\rangle_{k}$ defined by 
\begin{equation}
\label{eq:inner_product_H_k_times_k}
\langle \varphi,\psi\rangle_{k}=\langle \varphi^1,\psi^1\rangle_k + \langle \varphi^2, \psi^2 \rangle_k \, =\, \sum_{i=1}^2 \; \sum_{z\in\bb Z}\langle \varphi^i,h_z \rangle_0 \,  \langle \psi^i, h_z \rangle_0\,  \gamma_z^k.
\end{equation}
Analogously, we define the space $(\bb H_k \times \bb H_k)^\prime$ as  the dual of $\bb{H}_{k}\times \bb H_k$, relatively to the previous inner product on $(\bb L^2(\bb T))^2$. 
The inner product between two elements $\mc X^1, \mc X^2\in (\bb H_{k}\times\bb H_{k})^\prime$ is defined by
\begin{equation}
\label{eq:inner_product_H_minus_k_times_minus_k}
\langle \mc X^1,\mc X^2\rangle_{-k}=\sum_{z\in\bb Z} \left\{ \mc X^1 (h_z, 0)^\dagger \; \mc X^2(h_z,0)^\dagger \, +\, \mc X^1 (0,h_z)^\dagger\; \mc X^2(0,h_z)^\dagger\right\}\, \gamma_z^{-k}
\end{equation}
and we denote by $\|\cdot\|_{-k}$ the corresponding norm. 

\begin{remark}
\label{rem:topo_ident}
Observe that if $\mc Z=(\mc Z^1, \mc Z^2)^\dagger\in\bb H_{-k}\times \bb H_{-k}$ then the application $\mc Z\bigcdot$ given by \[\mc Z \bigcdot : (f^1,f^2)^\dagger \in \bb H_k \times \bb H_k \mapsto \mc Z^1 (f^1) + \mc Z^2 (f^2)\] is an element of $(\bb H_{k}\times\bb H_{k})^\prime$. 

Conversely, any element $\mc X \in (\bb H_{k}\times\bb H_{k})^\prime$ can be written in this form : take $\mc Z^1 (f)=\mc X (f,0)$ and $\mc Z^2 (f) =\mc X (0,f)$. In fact, the map $ \mc Z\mapsto \mc Z \bigcdot $ that we can define in this way, permits to identify topologically $ \bb H_{-k} \times {\bb H}_{-k}$ with $(\bb H_{k}\times\bb H_{k})^\prime$.
\end{remark}

\subsection{Fluctuation fields}

Fix  an integer $k$. Let us consider $0 <a \le 2$ and take $\theta(n)=n^a$. We recall that $\kappa>0$ and $\alpha_n=\alpha n^{-\kappa}$. Let us denote
\begin{equation}
\label{eq:velocity1}
c_n :=2 b^2 \rho\; \tfrac{\theta(n)\alpha_n}{n}= 2 b^2 \rho \alpha n^{a-\kappa -1}, \quad n\ge 1.
\end{equation} We also fix some horizon time $T>0$.  For any $n\ge 1$ we define the fluctuation field $\{ \mathcal{Y}^{n}_t\; ; \; t\in [0,T]\}$ for the variable $\xi$ as the random process living in the Skorokhod space $\bb D ([0,T],\bb H_{-k})$ such that 
\begin{equation}
\label{eq:xi_field}
\mathcal{Y}^{n}_t  (f) =\frac{1}{\sqrt{n}} \sum_{x\in \TT_n} \big(T^+_{c_n t}\, f \big)\big(\tfrac{x}{n}\big) \big(\xi_x(t\theta(n))-\rho \big), \quad f\in\mc D (\bb T).
\end{equation}
This means that we are looking at the fluctuation field of $\xi$ in a frame moving at velocity $c_n$. Similarly we define the fluctuation field $\{ \mathcal{V}^{n}_t\; ; \; t \in [0,T]\}$ for the variable $\eta$ as the random process living in the Skorokhod space $\bb D ([0,T], \bb H_{-k})$  such that
\begin{equation}
  \label{eq:vol_field}
  \mathcal{V}^{n}_t  (f) =\frac{1}{\sqrt{n}} \sum_{x\in \TT_n}  f\big(\tfrac{x}{n}\big)  \, \big(\eta_x(t\theta(n))-v \big), \quad f\in \mc D (\bb T).
\end{equation} 
In fact we are interested in the mutual evolution of these two fields, so we define
\begin{equation*}
\mc Z^n := 
\begin{pmatrix}
\mc Y^n\\
\mc V^n
\end{pmatrix}.
\end{equation*}
  Our aim is to study the convergence of the sequence $  (\mathcal{Z}^{n})_n$ according to the intensity of the asymmetry of the system which is regulated by the parameter $\kappa>0$.
%
%
Our main result is the following theorem. In what follows  $c=2b^2\rho\alpha$.

\begin{theorem}
We have that:
\begin{itemize}
\item If $\kappa>1$, then in the diffusive time scale $\theta (n)= n^2$,  the sequence of processes $(\mc Z^n)_n$ converges in law to $\mc Z \in \bb D ([0,T],(\mc D' (\TT))^2)$ which is the stationary solution of the Ornstein-Uhlenbeck equation \eqref{eq:ou_ini} with parameters:
\begin{equation}
\mathfrak{A}=
\gamma\, {\rm{I_2}}
\quad \textrm{and}\quad \mathfrak{C}=\gamma\begin{pmatrix}
\tau^2 & \delta  \\  \delta  & \sigma^2 \end{pmatrix}. 
\end{equation}

\item If $\kappa=1$, then in the diffusive time scale $\theta (n)= n^2$,  the sequence of processes $(\mc Z^n)_n$ converges in law to $\mc Z \in \bb D ([0,T],(\mc D '(\TT))^2)$ which is the stationary solution of the two-dimensional drifted Ornstein-Uhlenbeck equation \eqref{eq:2DOU} with parameters $\lambda=\mu=\gamma$,
$\theta=0$,  $\mf a=2\gamma\tau^2$, $\mf b=2\gamma\delta$ and $\mf d= 2\gamma\sigma^2$, and initial condition a two-dimensional Gaussian white noise (in space) with covariance matrix
\begin{equation*}
\begin{pmatrix}
\tau^2 & \delta  \\  \delta  & \sigma^2 \end{pmatrix}.
\end{equation*}
\item If $0 \le  \kappa <1$, then in the time scale $\theta (n)= n^{\kappa +1}$,  the sequence of processes $(\mc Z^n)_n$ converges in law to $\mc Z \in \bb D ([0,T], (\mc D' (\TT))^2)$ which is the stationary solution of the trivial transport equation given in \eqref{eq:trivial_0}  with parameter $\theta=-2\alpha b$, and initial condition a two-dimensional Gaussian white noise (in space) with covariance matrix
\begin{equation*}
\begin{pmatrix}
\tau^2 & \delta  \\  \delta  & \sigma^2 \end{pmatrix}.
\end{equation*}

\end{itemize}
Moreover in all these cases, if the time scale $\theta (n)=n^a$ is such that $a<\inf (2,\kappa+1)$ the evolution is trivial in the sense that the sequence of processes $(\mc Z^n)_n$ converges in law to $\mc Z_0$.
\end{theorem}

Hence this theorem fixes the minimal time scale needed in order to see a joint evolution of the fields of interest. It does not mean that this time scale is the only one for which a non trivial temporal evolution of the fields occurs. The next theorem shows that for the field $\mc Y^n$ we can go even further.

\begin{theorem}
\label{theo:xi_flu-braga}
The sequence of processes $(\mc Y^n)_n$ converges in law to $\mc Y$ which is an element of $\bb D ([0,T], \mc D' (\TT) )$ such that
\begin{itemize}
\item For any $\kappa>0$, if the time scale is $\theta (n) =n^a$ with $a<\inf(2, \frac{4}{3} (\kappa+1))$, then $d\mc Y_t=0$.
\medskip

\item If $\kappa>\frac12$, then in the diffusive time scale $\theta (n)=n^2$, $\mc Y$ is the stationary solution of the  Ornstein-Uhlenbeck equation{\footnote{or equivalently of \eqref{eq:sbe_ini} with $A=\gamma$, $B=0$ and $C=\gamma\tau$. }}  given by \eqref{eq:ou_ini}, in dimension $d=1$,  with ${\mf A}=\gamma$ and ${\mf C}=\gamma \tau$.

\medskip

\item If $\kappa=\frac 12$, then in the diffusive time scale $\theta(n)=n^2$, $\mc Y$ is the stationary energy solution of the one-dimensional stochastic Burgers equation \eqref{eq:sbe_ini} with parameters $A=\gamma$, $B=b^2\alpha$ and $C=\gamma\tau$.

\end{itemize}
\end{theorem}

We conjecture that the result in the first item is true for any $a<2$ and $\kappa>a-\frac32$ but below we use   Theorem 4 of \cite{BG}  which is not optimal in this case. The behavior when  $b=a-\frac32$ is open.

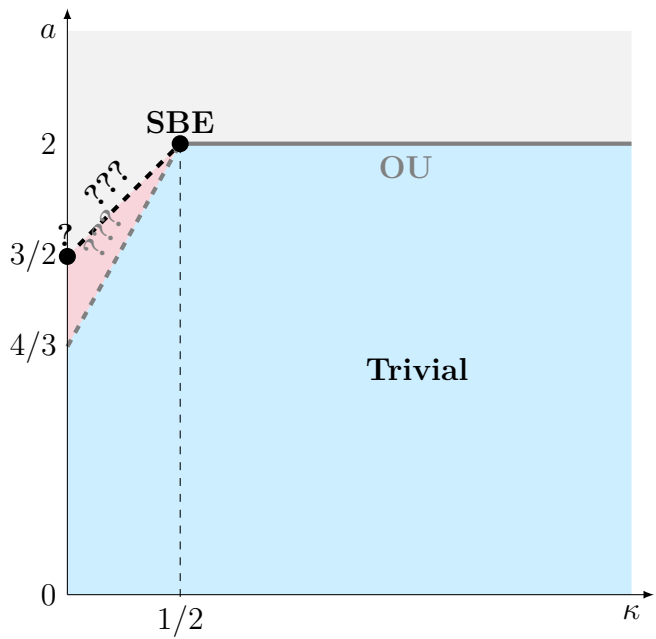
\begin{figure}[h!]
\begin{center}
\begin{tikzpicture}[scale=0.3]
\draw (0,25) node[left]{$a$};
\draw (25,0) node[below]{$\kappa$};
\draw (5,0) node[below]{$1/2$};
\draw (0,0) node[left]{$0$};
\draw (0,15) node[left]{$3/2$};
\draw (0,11) node[left]{$4/3$};
\draw (0,20) node[left]{$2$};
\fill[light-gray] (0,11) -- (5,20) -- (25,20) -- (25,25) -- (0,25) -- cycle;
\fill[fill=columbiablue, fill opacity=0.5] (0,0) -- (25,0) -- (25,20) -- (5,20) -- (0,11)--cycle;
\fill[fill=cherryblossompink, fill opacity=0.5] (0,11) -- (5,20) -- (0,15)--cycle;
\draw[-,=latex,gray,ultra thick] (5,20) -- (25, 20) node[midway,below,sloped] {\bf OU};
\draw[-,=latex,black,dashed, ultra thick] (0,15) -- (5,20) node[midway,above,sloped] {\bf ???};
\draw[-,=latex,gray,dashed, ultra thick] (0,11) -- (5,20) node[midway,above,sloped] {\bf ???};
\node[circle,fill=black,inner sep=0.8mm] at (5,20) {};
\node[circle,fill=black,inner sep=0.8mm] at (0,15) {};
\node[] at (-0.1,15) [above] {\bf ?};
\node[] at (5,20) [above] {\bf SBE};
\node[] at (15.5,10) {\bf Trivial};
\draw[-,=latex, dashed] (5,-0.1) -- (5,20);
\draw[->,>=latex] (0,0) -- (26,0);

\draw[->,>=latex] (0,0) -- (0,26);

\end{tikzpicture}
\end{center}
\caption{$\xi_x$ fluctuations}
\label{fig:xi}
\end{figure}

%
%
%
%
%

\begin{figure}[h!]
\begin{center}
\begin{tikzpicture}[scale=0.3]
\draw (0,25) node[left]{$a$};
\draw (25,0) node[below]{$\kappa$};
\draw (10,0) node[below]{$1$};
\draw (5,0) node[below]{$1/2$};
\draw (0,0) node[left]{$0$};
\draw (0,10) node[left]{$1$};
\draw (0,20) node[left]{$2$};
\fill[light-gray] (0,10) -- (10,20) -- (25,20) -- (25,25) -- (0,25) -- cycle;
\fill[fill=columbiablue, fill opacity=0.5] (0,0) -- (25,0) -- (25,20) -- (10,20) -- (0,10)--cycle;
\draw[-,=latex,gray,ultra thick] (10,20) -- (25, 20) node[midway,below,sloped] {\bf OU};
\draw[-,=latex,cherryblossompink, ultra thick] (0,10) -- (10,20) node[midway,above,sloped] {$2d$-\bf drifted OU};
\node[circle,fill=cherryblossompink,inner sep=0.8mm] at (10,20) {*};
\node[] at (10,20) [above] {};
\node[] at (15.5,10) {\bf Trivial};
\draw[-,=latex, dashed] (10,-0.1) -- (10,10);
\draw[-,=latex, dashed] (10,-0.1) -- (10,20);
\draw[-,=latex, dashed] (5,-0.1) -- (5,15);
\draw[->,>=latex] (0,0) -- (26,0);
\draw[->,>=latex] (0,0) -- (0,26);

A
\end{tikzpicture}
\end{center}
\caption{Joint fluctuations}
\label{fig:volume}
\end{figure}
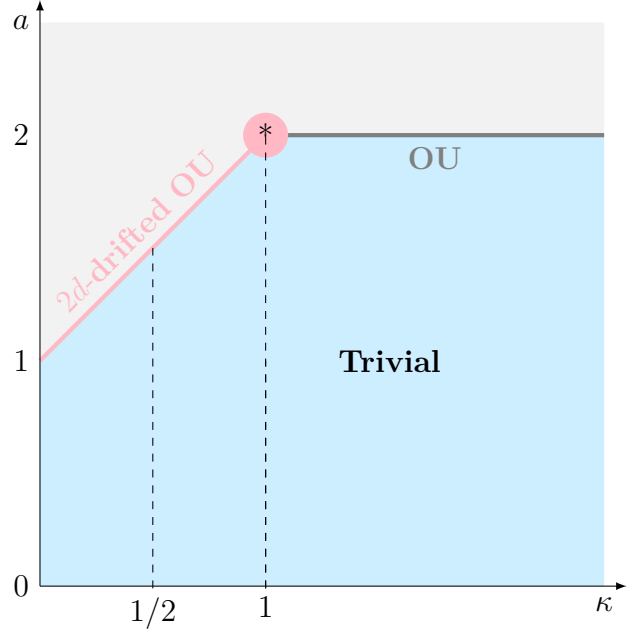

%
%
%
%

\section{Sketch of the proof of the main theorems} \label{sec:sketch}

For any $d \ge 1$ and $n\ge 1$ and any function $u:\TT \to \RR^d$ the discrete gradient $\nabla_n u$ (resp. Laplacian $\Delta_n u $)  is the function defined on $\tfrac{1}{n} \TT_n$ by
\begin{equation*}
\begin{split}
&(\nabla_n u)\,  \big(\tfrac{x}{n}\big)= n\left[ u\big(\tfrac{x+1}{n}\big) - u \big(\tfrac{x}{n}\big)\right], \\
&(\Delta_n u)\,   \big(\tfrac{x}{n}\big)= n^2 \left[ u \big(\tfrac{x+1}{n}\big)+ u\big(\tfrac{x-1}{n}\big) -2 u\big(\tfrac{x}{n}\big)\right], \quad x \in \TT_n.
\end{split}
\end{equation*}
Along the proofs we will use frequently the following bound based on the Cauchy-Schwarz inequality and stationarity of the process. We recall that $\langle \cdot \rangle$ denotes the average with respect to the equilibrium measure $\mu_{\bar \beta, \bar \lambda}$. If $F: [0,T]\times \Omega_n \to \RR$ is a function such that $\int_0^T \langle F^2 (s, \cdot) \rangle \, ds <\infty$ then we have
\begin{equation}
\label{eq:CSs}
\forall t\in [0,T], \quad \EE \bigg[ \Big( \int_0^t F(s, \eta(s))\, ds  \Big)^2\bigg] \le t \int_0^t \left\langle F^2 (s, \cdot) \right\rangle \, ds. 
\end{equation}
Observe that the r.h.s. of \eqref{eq:CSs} is usually easy to compute or estimate since it involves only a static expectation while the l.h.s. involves a dynamical expectation. 

\medskip

It turns out convenient to introduce the \emph{mutual} field $\mc X^n:=\mc Z^n \bigcdot$ defined by 
\begin{align}
  \mathcal{Z}^{n}_t  \bigcdot \vec f &=\cfrac{1}{\sqrt{n}} \sum_{x\in \TT_n}\Big\{
  \big(T^+_{c_n t}f^1 \big)\big(\tfrac{x}{n}\big)  \left(  \xi_x(t\theta(n)-\rho \right)+f^2\big(\tfrac{x}{n}\big)\left( \eta_x(t\theta(n))-v \right)\Big\} \notag \\  &=  \mathcal{Y}^{n}_t  (f^1)+ \mathcal{V}^{n}_t  (f^2) \vphantom{\int} \label{eq:chi_field}
\end{align} 
where $\vec f=(f^1,f^2)^\dagger \in \mc D(\bb T,\bb R^2)$. Then, the fluctuation field $\{\mc X^n_t\; ; \; t\in[0,T]\}$ is an element of the Skorokhod space $\bb D ([0,T], (\bb H_{k}\times \bb H_{k})^\prime)$.

\subsection{Characterization of limit points for the $\mc Z$ field}
\label{sec:char:X-field}

In Section \ref{sec:tightness_chi} we will prove that in a certain range of time scales the sequence of processes $(\mc X^n)_n$ is tight in $\bb D([0,T], (\bb H_k\times \bb H_k)')$, for some $k$. This implies, by Remark \ref{rem:topo_ident}, that $(\mc Z^n)_n$ is also tight in $\bb D([0,T],\bb H_{-k}\times \bb H_{-k})$. Therefore, up to a subsequence, we may assume that the sequences above converge in the respective spaces. The results of this section are restricted to the following range of time scales:
\begin{equation*}
a \le \inf (\kappa +1, 2).
\end{equation*}
In particular, in the regime $\kappa \in (0,1)$ we are not able to study the limit of the sequence $(\mc X^n)_n$ if the parameter $a$ of the time scale $n^a$ is strictly bigger than the \textit{transition line} $a=\kappa+1$.

\medskip

For $\vec f=(f^1, f^2)^\dagger \in\mc D (\bb T, \bb R^2)$ and $t\in [0,T]$ we define
\begin{equation}
\label{eq:Dynkin_eta}
\mc N_t^n\bigcdot \vec f =  \mathcal{Z}^{n}_t  \bigcdot \vec f -  \mathcal{Z}^{n}_0  \bigcdot \vec f -\int_0^t(\partial_s+\theta(n)\mc L  ) \, (\mathcal{Z}^{n}_s \bigcdot \vec f) \, ds.
\end{equation}
By  Dynkin's formula $\{\mc N_t^n\bigcdot \vec f  \; ;\; t\in[0,T]\}$ is a martingale. We have
\begin{align}
(\partial_s+\theta(n)\mc L  ) \,& ( \mathcal{Z}^{n}_s  \bigcdot \vec f)   =\frac{\gamma \theta (n)}{n^{2}} \, \mc Z^n_s \bigcdot \Delta_n \vec f \label{eq:mart_dec_chi1}\\
&+
b \;\frac{\theta (n)\alpha_n}{n^{3/2}}\sum_{x\in \TT_n} \big(\nabla_n f^2 \big)\big(\tfrac{x}{n}\big) \, (\xi_x(s\theta(n))+\xi_{x+1}(s\theta(n))\label{eq:mart_dec_chi2}\\
&-b^2 \; \frac{\theta(n) \alpha_n}{n^{3/2}}\sum_{x\in\bb T_n} \big(\nabla_n T^+_{c_n s} \, f^1\big) \big(\tfrac xn\big) \, \bar{\xi}_x(s\theta(n))\bar{\xi}_{x+1}(s\theta(n)). \label{eq:mart_dec_chi3}\end{align}
Above and in what follows, for a random variable $X$, the random variable $\bar X$ denotes the centered variable $X-\EE[X]$. Recall that $\theta(n)=n^a$ and $\alpha_n=\alpha n^{-\kappa}$. Note that by using \eqref{eq:CSs} it is easy to check that the variance of the time integral of \eqref{eq:mart_dec_chi1} has variance of order $\mathcal{O}(\theta(n)^2 \; n^{-4})$, which vanishes if $a<2$. 
Now,  \eqref{eq:mart_dec_chi2} can be rewritten as
\begin{equation}
\label{eq:cacaboudin}
2b \; \frac{\theta (n)\alpha_n}{n^{3/2}}\sum_{x\in \TT_n} \big(\nabla_n f^2 \big) \big(\tfrac{x}{n}\big) \, \bar{\xi}_x(s\theta(n))+ b \; \frac{\theta (n)\alpha_n}{n^{5/2}}\sum_{x\in \TT_n} \big( \Delta_n f^2 \big) \big(\tfrac{x}{n}\big)\, \bar{\xi}_x(s\theta(n)).
\end{equation}
As above, by \eqref{eq:CSs}, the time integral of the term at the r.h.s.~of last expression has variance of order $\mathcal{O}(\alpha_n^2\; \theta(n)^2\; n^{-4})$ while the remaining  term in \eqref{eq:cacaboudin} can be written as 
\begin{equation*}
2b\; \frac{\theta (n)\alpha_n}{n^{3/2}}\sum_{x\in \bb T_n}\big( \nabla_n f^2\big)\big(\tfrac{x}{n}\big) \bar{\xi}_x(s\theta(n))=2b \; \frac{\theta (n)\alpha_n}{n} \; \mc Y^n_{s}\big(\nabla_n T_{c_n s}^-\,  f^2\big).
\end{equation*}
Note that by \eqref{eq:CSs} the variance of the time integral of the term of last expression is bounded from above by $\mc O(\alpha_n^2 \; \theta(n)^2\; n^{-2}).$ This means that when $a<\kappa +1$ that term does not contribute to the limit. The time integral of  \eqref{eq:mart_dec_chi3} has a variance bounded from above (use again \eqref{eq:CSs}) by 
\begin{equation}
\label{eq:termSBE}
\EE \bigg[ \Big( b^2\frac{\theta(n) \alpha_n}{n^{3/2}} \int_0^t \sum_{x\in\bb T_n} \big(\nabla_n T^+_{c_n s} \, f^1\big) (\tfrac xn) \, \bar{\xi}_x(s\theta(n))\bar{\xi}_{x+1}(s\theta(n))  \, ds\Big)^2\bigg] \lesssim \cfrac{\theta(n)^2 \alpha_n^2}{n^2}
\end{equation}
which vanishes if $a<\kappa+1$. In fact, we will show in Section \ref{sec:char:Y-field} that a less rough estimate than \eqref{eq:CSs} shows, in fact, that the expectation in \eqref{eq:termSBE}  is of order $\mathcal{O}(\alpha_n^2 \; (\theta (n))^{3/2}\; n^{-2})$, so that it goes to $0$ as soon as $a<\tfrac{4}{3}(\kappa +1)$.

\medskip

Let us now study the quadratic variation of the martingale term in \eqref{eq:Dynkin_eta}. By Dynkin's formula, the quadratic variation of the martingale $\mc N_t^n\bigcdot \vec f $ is given by
\begin{align}
\langle\mc N^n\bigcdot \vec f \, \rangle_t& =  \int_0^t \left\{ \theta(n)\mc L (( \mathcal{Z}^{n}_s  \bigcdot \vec f)^2)-2  \theta (n) (\mathcal{Z}^{n}_s  \bigcdot \vec f ) \, \mc L  (\mathcal{Z}^{n}_s  \bigcdot \vec f) \right\}\,  ds \notag \\
&= \gamma \theta (n) \int_0^t \left\{ \mc S (( \mathcal{Z}^{n}_s  \bigcdot \vec f)^2)-2 (\mathcal{Z}^{n}_s  \bigcdot \vec f ) \,  \mc S (  \mathcal{Z}^{n}_s  \bigcdot \vec f) \right\}\,  ds,
\label{eq:qv_chi} \end{align}
and a simple computation shows that last display is equal to
\begin{align}
&\langle\mc N^n  \bigcdot \vec f \, \rangle_t=\frac{\gamma \theta(n)}{n^3} \int_0^t \sum_{x\in\bb T_n}\big(\nabla_n \big(T^+_{c_n s}f^1 \big)\big)^2 \big(\tfrac{x}{n}\big) \,\big[\xi_{x+1}(\theta(n)s)-\xi_x(\theta(n)s)\big]^2\, ds \notag \\
&+\frac{\gamma \theta(n)}{n^3} \int_0^t \sum_{x\in\bb T_n}\big(\nabla_n f^2\big)^2 \big(\tfrac{x}{n}\big) \,\big[\eta_{x+1}(\theta(n)s)-\eta_x(\theta(n)s)\big]^2\, ds \notag \\
&+2\frac{\gamma \theta(n)}{n^3} \int_0^t \sum_{x\in\bb T_n}\big(\nabla_n \big(T^+_{c_n s}f^1 \big)\big) \big(\tfrac{x}{n}\big) \big(\nabla_n f^2\big)\big(\tfrac{x}{n}\big) \,\big[\xi_{x+1}(\theta(n)s)-\xi_x(\theta(n)s)\big] \notag \\
& \hspace{5cm} \times\big[\eta_{x+1}(\theta(n)s)-\eta_x(\theta(n)s)\big] \, ds. \label{eq:qv_mart_chi_field_2}
\end{align}
If $a<2$ then the $\bb L^1$-norm of the quadratic variation of $\mc N^n  \bigcdot \vec f$ vanishes as $n\to+\infty$. If $a=2$ then we have that
\begin{equation*}
\begin{split}
\bb E \left[  \langle\mc N^n  \bigcdot \vec f \, \rangle_t \right] &=2 \gamma \tau^2 \; \cfrac{1}{n}  \sum_{x\in\bb T_n} \int_0^t \big(\nabla_n \big(T^+_{c_n s}f^1 \big)\big)^2\big(\tfrac{x}{n}\big) \, ds \\& \qquad +  2 \gamma \sigma^2 t \; \cfrac{1}{n} \sum_{x\in\bb T_n}\big(\nabla_n f^2\big)^2 \big(\tfrac{x}{n}\big)\\
& \qquad + 4\gamma \delta \int_0^t \cfrac{1}{n} \sum_{x\in\bb T_n}\big(\nabla_n \big(T^+_{c_n s}f^1 \big)\big) \big(\tfrac{x}{n}\big) \big(\nabla_n f^2\big) \big(\tfrac{x}{n}\big) \, ds. 
\end{split}
\end{equation*}
Recall that $c_n = 2b^2 \rho \alpha n^{a-\kappa -1}=2b^2 \rho \alpha n^{1-\kappa}$. It follows that 
\begin{itemize}
\item If $\kappa <1$ then $c_n \to \infty$ and therefore
\begin{equation*}
\lim_{n \to \infty} \bb E \left[ \langle\mc N^n  \bigcdot \vec f \, \rangle_t\right] = 2\gamma \tau^2  t\int_\TT ( \nabla f^1 )^2 (y) \, dy + 2\gamma \sigma^2 t  \int_\TT ( \nabla f^2)^2 (y) \, dy. 
\end{equation*}
\item If $\kappa =1$ then $c_n= c:=2 b^2 \rho \alpha$ and therefore
\begin{equation*}
\begin{split}
\lim_{n \to \infty}  \bb E \left[ \langle\mc N^n  \bigcdot \vec f \, \rangle_t\right]  &= 2\gamma \tau^2 t \int_\TT ( \nabla f^1 )^2 (y) \, dy + 2\gamma \sigma^2 t \int_\TT ( \nabla f^2 )^2 (y) \, dy\\
& \quad + 4\gamma \delta \int_\TT c^{-1} \big(f^1(y+ct)- f^1 (y)\big) \, ( \nabla f^2) (y) \, dy. \end{split}
\end{equation*}
Note that last expression is equal to 
\begin{equation*}
\begin{split}
\lim_{n \to \infty}  \bb E \left[ \langle\mc N^n  \bigcdot \vec f \, \rangle_t\right] &= 2\gamma \tau^2 t \int_\TT ( \nabla f^1 )^2 (y) \, dy + 2\gamma \sigma^2 t \int_\TT ( \nabla f^2 )^2 (y) \, dy\\
& \quad + 4\gamma \delta \int_0^t \int_\TT \big(\nabla T^+_{cs}f^1 \big)(y) \, ( \nabla f^2 ) (y) \, dy ds. 
\end{split}
\end{equation*}

\item If $\kappa >1$ then $c_n \to 0$ and therefore
\begin{equation*}
\begin{split}
\lim_{n \to \infty}  \bb E \left[ \langle\mc N^n  \bigcdot \vec f \, \rangle_t\right] &= 2\gamma \tau^2 t \int_\TT ( \nabla f^1 )^2 (y) \, dy + 2\gamma \sigma^2 t \int_\TT ( \nabla f^2 )^2 (y) \, dy\\
&\quad + 4\gamma \delta \int_\TT (\nabla f^1 ) (y)  \, ( \nabla f^2 ) (y) \, dy. 
\end{split}
\end{equation*}

\end{itemize}

We have then 
\begin{itemize}
\item If $\kappa>1$ and $a=2$ then
\[
\mc N^n_t \bigcdot \vec f = \mathcal{Z}^{n}_t  \bigcdot \vec f - \mathcal{Z}^{n}_0  \bigcdot \vec f+\gamma \int_0^t \mc Z^n_s \bigcdot \Delta_n \vec f \; ds
\]
plus terms which vanish as $n\to\infty$ in the $\bb L^2$-norm.  Moreover, the quadratic variation of the martingale satisfies:
\begin{equation*}
\lim_{n\to+\infty}\bb E [\langle\mc N^n  \bigcdot \vec f \, \rangle_t]= 2t \tau^2\gamma\|\nabla f^1 \|_0^2+2t \sigma^2\gamma \|\nabla f^2 \|_0^2 + 4\gamma \delta t \langle \nabla f^1 , \nabla f^2  \rangle_0.
\end{equation*} 
Then, the limiting field $\{\mc Z_t\; ; \; t\in [0,T]\}$ satisfies
\[
\mc N_t \bigcdot \vec f = \mathcal{Z}_t  \bigcdot \vec f - \mathcal{Z}_0  \bigcdot \vec f+\gamma \int_0^t \mc Z_s \bigcdot \Delta_n \vec f \; ds
\]
 so that it is a solution of \eqref{eq:ou_ini} as given in Definition \ref{def:uni_OU}, with 
 \begin{equation}
{\mf A} =
\left( \begin{array}{cc}
 \gamma & 0 \\ 0&\gamma \end{array}\right), \qquad
 \label{eq:matrix_D_ini}
\mf D=\begin{pmatrix}
\tau^2 & \delta \\ \delta & \sigma^2 \end{pmatrix}, \qquad \mf C=\gamma \mf D. 
\end{equation}

\item If $\kappa>a-1$ and $a<2$ then 
\[
\mc N^n_t  \bigcdot \vec f= \mathcal{Z}^{n}_t   \bigcdot \vec f - \mathcal{Z}^{n}_0  \bigcdot \vec f
\]
plus terms which vanish as $n\to\infty$ in in the $\bb L^2$-norm. 
Moreover, the quadratic variation of the martingale satisfies:
\begin{equation*}
\lim_{n\to+\infty}\bb E[\langle\mc N^n \bigcdot \vec f \rangle_t]= 0.
\end{equation*}
Then $\{\mc Z_t\; ; \; t\in [0,T]\}$ has a trivial evolution given by: 
\begin{equation} 
\label{eq:trivial_vol}
\mc Z_t  \bigcdot \vec f =\mc Z_0  \bigcdot \vec f, \quad\textrm{so that }  d\mc Z_t=0.
\end{equation}
\item If $\kappa=a-1$ and $a<2$ then 
\begin{equation}
\mc N^n_t  \bigcdot \vec f = \mathcal{Z}^{n}_t   \bigcdot \vec f - \mathcal{Z}^{n}_0  \bigcdot \vec f -2b \tfrac{\theta (n)\alpha_n}{n} \int_0^t\mc Z^n_{s}\bigcdot \big(\nabla_n T_{c_n s}^- \; f^2\, , \, 0 \big)^\dagger \, ds
\end{equation}
is equal to 
\begin{equation}
\label{eq:new_exp}
\mc N^n_t  \bigcdot \vec f = \mathcal{Z}^{n}_t   \bigcdot \vec f - \mathcal{Z}^{n}_0  \bigcdot \vec f -2b\alpha\int_0^t \mc Z^n_{s}\bigcdot \big (\nabla_n T_{c s}^- \; f^2\, , \, 0 \big)^\dagger \, ds
\end{equation}
where we recall that $c= 2b^2 \rho \alpha$. Moreover, the quadratic variation of the martingale satisfies:
\begin{equation*}
\lim_{n\to+\infty}\bb E [\langle\mc N^n  \bigcdot \vec f \, \rangle_t]= 0.
\end{equation*}
Therefore, by  Lemma \ref{lem:conv-int-y} (applied with $\gamma=0$), we have that the limiting field $\{\mc Z_t\; ; \; t\in [0,T]\}$ satisfies
\begin{equation*}
\mathcal{Z}_t   \bigcdot \vec f  = \mathcal{Z}_0  \bigcdot \vec f  -2b\alpha\int_0^t\mc Z_{s}\bigcdot  \big(\nabla T_{cs}^- f^2\, ,\, 0\big)^\dagger \, ds.
\end{equation*}
Therefore, $\{\mc Z_t\; ; \; t\in [0,T]\}$ is the solution of the trivial transport equation as defined in Remark \ref{rem:trivial} with $\theta=-2\alpha b$, \textit{i.e.}
\begin{equation}\label{eq:op_1}
{\mf L}_t =
\left( \begin{array}{cc}
 0 & 0 \\ -2\alpha b\nabla T_{ct}^+&0 \end{array}\right),
\end{equation} and initial condition a centered Gaussian field with covariance matrix $\mf D$ as given in \eqref{eq:matrix_D_ini}.

\item If  $\kappa=1$ and $a=2$ then we have
\begin{equation}
\label{eq:new_exp_diff}
\begin{split}
\mc N^n_t  \bigcdot \vec f= \mathcal{Z}^{n}_t   \bigcdot \vec f - \mathcal{Z}^{n}_0  \bigcdot \vec f &+\gamma \int_0^t \mc Z^n_s \bigcdot \Delta_n \vec f \, ds \\ & -2b\alpha\int_0^{t}\mc Y^n_{s}(\nabla_n T_{c s}^-f^2)\,  ds,
\end{split} \end{equation}
which is equal to 
\begin{equation}
\label{eq:new_exp_diff_1}
\begin{split}
\mc N^n_t  \bigcdot \vec f = \mathcal{Z}^{n}_t   \bigcdot \vec f - \mathcal{Z}^{n}_0  \bigcdot \vec f &+\gamma \int_0^t \mc Z^n_s \bigcdot \Delta_n \vec f \, ds \\ & -2b\alpha\int_0^{t}\mc Z^n_{s}\bigcdot (\nabla_n T_{c s}^-f^2,0)^\dagger \, ds.
\end{split}
\end{equation}
Moreover, the quadratic variation of the martingale satisfies:
\begin{multline*}
\lim_{n\to+\infty}\bb E[\langle\mc N^n  \bigcdot \vec f \, \rangle_t] \\ = 2t \tau^2\gamma\|\nabla f^1\|_0^2+2t \sigma^2\gamma\|\nabla f^2 \|_0^2 + 4\gamma \delta c^{-1} \, \big\langle T^+_{ct} f^1 -f^1 \, ,\,  \nabla f^2 \big\rangle_0.
\end{multline*} 
Then, by Lemma \ref{lem:conv-int-y} (applied with $\alpha=0$) and Corollary \ref{martconv}, we conclude that the sequence $\{\mc Z^n_t\; ;\; t\in [0,T]\}_n$ converges to the solution of equation \eqref{eq:2DOU} as given in Definition \ref{def:ed}  with $\mf L_t$ as in \eqref{eq:operator_L_t} with $\lambda=\mu=\gamma$ and $\theta=0$:
\begin{equation}\label{eq:op_2}
{\mf L}_t=\gamma
\left( \begin{array}{cc}
  \Delta & 0 \\ 0 & \Delta \end{array}\right),
\end{equation}
and ${\mf C}_t $ as in \eqref{eq:C_t} with $\mf a=2\gamma\tau^2$, $\mf b=2\gamma\delta$ and $\mf d= 2\gamma\sigma^2$, \textit{i.e.}
\begin{equation}\label{eq:C_t_new}
{\mf C}_t =\gamma
\left( 
\begin{array}{cc}
 2\tau^2 & 2\delta T^-_{ct}\\
2\delta  T^+_{ct} & 2\sigma^2 
\end{array}
\right),
\end{equation}
 and $\mf D$ as in \eqref{eq:matrix_D_ini}.

\end{itemize}

\begin{lemma}
\label{lem:conv-int-y}
Let $\gamma, \alpha \ge 0$ and fix $\vec f =(f^1,f^2)^\dagger \in \bb H_k \times \bb H_k$.  Assume that the sequence of processes $\{\mc Z_t^n \; ;\; t\in [0,T]\}_n$ converges in law, as a process of $\bb D ([0,T]\, ,\, ({\bb  H}_{k} \times {\bb H}_{k})^\prime)$, to $\{\mc Z_t \; ;\; t\in [0,T]\}$. Then, for any $t\in [0,T]$, the sequence of random variables
\begin{multline}
\label{eq:tresor}
\big\{ \texttt{Z}_n \big\}_n\\
 :=\Big\{  \mathcal{Z}^{n}_t   \bigcdot \vec f - \mathcal{Z}^{n}_0  \bigcdot \vec f +\gamma \int_0^t \mc Z^n_s \bigcdot \Delta_n \vec f \, ds  -2b\alpha\int_0^{t}\mc Z^n_{s}\bigcdot (\nabla_n T_{c s}^-f^2,0)^\dagger \, ds\Big\}_n 
\end{multline}
converges, as $n\to +\infty$, to the random variable
\begin{equation*}
\texttt{Z}:= \mathcal{Z}_t   \bigcdot \vec f - \mathcal{Z}_0  \bigcdot \vec f +\gamma \int_0^t \mc Z_s \bigcdot \Delta \vec f \, ds -2b\alpha\int_0^{t}\mc Z_{s}\bigcdot (\nabla T_{c s}^-f^2,0)^\dagger \, ds.
\end{equation*}
\end{lemma}
\begin{proof}
By performing a Taylor expansion on $\vec f$ we can replace in the expression (\ref{eq:tresor}), the discrete gradient and Laplacian by the continuous gradient and Laplacian, up to terms which vanish in $\bb L^2$ (use (\ref{eq:CSs})):
\begin{equation*}
\lim_{n \to \infty} \bb E \left[ \left( \texttt{Z}_n -\texttt{Z}_n^\prime \right)^2\right]=0
\end{equation*}
where 
\begin{equation}
\label{eq:tresor2}
\texttt{Z}_n^\prime = \mathcal{Z}^{n}_t   \bigcdot \vec f - \mathcal{Z}^{n}_0  \bigcdot \vec f +\gamma \int_0^t \mc Z^n_s \bigcdot \Delta \vec f \, ds  -2b\alpha\int_0^{t}\mc Z^n_{s}\bigcdot (\nabla  T_{c s}^-f^2,0)^\dagger \, ds.
\end{equation}
Let us first analyse the last term on the l.h.s.~of (\ref{eq:tresor2}). We split the time integral on $[0,t]$ in a sum of time integrals on intervals of size $th$, $h>0$ being small, as:
\begin{equation}
\int_0^{t}\mc Z^n_{s}\bigcdot (\nabla  T_{c s}^-f^2,0)^\dagger \, dr = \int_0^t\mc Y^n_{s}(\nabla T_{c s}^- \, f^2)\, ds= \sum_{k=0}^{1/h} \int_{{kth}}^{(k+1)t h}\mc Y^n_{s}(\nabla T_{c s}^- \,f^2)\, ds.
\end{equation}
Without loss of generality we can assume that $1/h$ is an integer. Now in each integral of the r.h.s.~we can sum and subtract $\nabla T_{c t k h }^- \, f^2$ inside, so that, by linearity of $\mc Y^n_s$,  last display is equal to
\begin{equation}
 \sum_{k=0}^{1/h} \int_{{k th}}^{(k+1) th}\mc Y^n_{s} (\nabla T_{c s}^-\,  f^2 -\nabla T_{c t k h}^- \, f^2)\, ds\; + \; \sum_{k=0}^{1/h} \int_{{k th}}^{(k+1) t h}\mc Y^n_{s}(\nabla T_{c t k h}^- \, f^2)\, ds.
\end{equation}
Now we estimate the $\mathbb L^2$-norm of the term at the l.h.s.~of last display.  
From Minkowski's inequality and \eqref{eq:CSs},  we have that 
\begin{equation*}\begin{split}
 &\sqrt{\bb E \bigg[\Big(\sum_{k=0}^{1/h} \int_{{k th}}^{(k+1)t h}\mc Y^n_{s}\big(\nabla T_{cs}^- \, f^2 -\nabla T_{c k th}^- \, f^2  \big) \, ds\Big)^2\bigg]}\\
 &\leq \sum_{k=0}^{1/h}\sqrt{\bb E \bigg[\Big(\int_{{k th}}^{(k+1)t h}\mc Y^n_{s}\big(\nabla T_{c s}^- \, f^2 -\nabla T_{c  k t h}^- \, f^2 \big)\, ds\Big)^2\bigg]}\\
 & \leq \sum_{k=0}^{1/h}\sqrt{\bb E \left[ th \; \int_{{k th}}^{(k+1)t h}\left(\mc Y^n_{s}\big(\nabla T_{c s}^- \, f^2 -\nabla T_{c k th}^- \, f^2 \big)\right)^2 \, ds\right]}\\
 &= \sum_{k=0}^{1/h}\sqrt{th}\; \sqrt{\int_{k th}^{(k+1)t h}\EE  \left[ \left(\mc Y^n_{0}\big(\nabla T_{c s}^- \, f^2 -\nabla T_{c k th}^- \, f^2 \big)\right)^2\right]  \, ds}\\
 &\lesssim \sum_{k=0}^{1/h}\sqrt{th}\; \left(\frac{1}{n}\sum_{x\in\bb T_n} \int_{k th}^{( k+1)th}\big(\nabla T_{c s}^-\,  f^2 -\nabla T_{c k th}^- f^2 \big)^2 \big(\tfrac{x}{n}\big) \, ds \right)^{\frac12},
\end{split}\end{equation*}
where the last inequality uses an explicit computation with the initial distribution. 
By Taylor expansion on the function $f^2$, we can bound the last term from above by
 \begin{equation}
 C(f^2,T)\sum_{k=0}^{1/h} (th)^2\leq C(f^2,T) h,
\end{equation}
which vanishes as $h\to 0$. From the last computations, we are able to rewrite \eqref{eq:new_exp}
 as
\begin{equation*}
\begin{split}
\int_0^t\mc Y^n_{s}(\nabla T_{c s}^- \, f^2)\, ds\; & =\; \sum_{k=0}^{1/h}\int_{{k th}}^{(k+1)t h}\mc Y^n_{s}(\nabla T_{c k th}^- \, f^2)\, ds \; + \varepsilon_{n,h} \\
&\; =\; \sum_{k=0}^{1/h}\int_{{k th}}^{(k+1)t h}\mc Z^n_{s} \bigcdot (\nabla T_{c k th}^- \, f^2\, , \, 0)^\dagger \, ds \; + \varepsilon_{n,h}
\end{split}
\end{equation*}
where the term $\varepsilon_{n,h}$ vanishes as $n\to\infty$ and then $h\to 0$ in $\bb L^2$. Similarly we have that
\begin{equation}
\int_0^t\mc Y_{s}(\nabla T_{c s}^- \, f^2)\, ds\;  =\; \sum_{k=0}^{1/h}\int_{{k th}}^{(k+1)t h}\mc Z_{s}\bigcdot (\nabla T_{c k th}^- \, f^2\, , \, 0)^\dagger \, ds \; + \varepsilon_{h}
\end{equation}
where the term $\varepsilon_{h}$ vanishes as $h\to 0$ in $\bb L^2$. Therefore it is sufficient to prove that for each fixed $h>0$, the sequence of random variables 
\begin{equation*}
\left\{  \mathcal{Z}^{n}_t \bigcdot \vec f- \mathcal{Z}^{n}_0 \bigcdot \vec f+ \gamma \int_0^t \mc Z^n_s \bigcdot \Delta \vec f ds -2b\alpha \sum_{k=0}^{1/h}\int_{{k th}}^{(k+1)t h}\mc Z^n_{s}\bigcdot (\nabla T_{c k th}^- \, f^2\, , \, 0)^\dagger \, ds  \right\}_{n} 
\end{equation*}
converges in law to the random variable
\begin{equation*}
\mathcal{Z}_t \bigcdot \vec f- \mathcal{Z}_0 \bigcdot \vec f+ \gamma \int_0^t \mc Z_s \bigcdot \Delta \vec f ds -2b\alpha \sum_{k=0}^{1/h}\int_{{k th}}^{(k+1)t h}\mc Z_{s}\bigcdot (\nabla T_{c k th}^- \, f^2\, , \, 0)^\dagger \, ds.
\end{equation*}
For a fixed $h>0$ and a fixed $t \in [0,T]$, the application 
\begin{equation*}
\begin{split}
\mc Z \in {\bb D} ([0,T], ({\bb H}_{k} \times{\bb H}_{k})^\prime ) \mapsto &\; \mathcal{Z}_t  \bigcdot \vec f - \mathcal{Z}_0 \bigcdot \vec f + \gamma \int_0^t \mc Z_s \bigcdot \Delta \vec f ds\\
& -2b\alpha \sum_{k=0}^{1/h}\int_{{k th}}^{(k+1)t h}\mc Z_{s} \bigcdot (\nabla T_{c k th}^- \, f^2 , 0)^\dagger \, ds \in \RR
\end{split}
\end{equation*} 
is continuous. Therefore the result becomes a trivial consequence of the assumption.
%
%
\end{proof}

\subsection{Characterization of limit points for the $\mc Y$ field}
\label{sec:char:Y-field}
The results of the previous section are restricted to the range of time scales $a<\inf(\kappa+1,2)$. In this section, we show we can go beyond this range but only for the sequence $(\mc Y^n)_n$. In the range of time scales considered below we show in Section \ref{sec:tightness_xi} that the sequence $( \mc Y^n)_n$ is tight. Therefore we may assume (up to a subsequence) that it is converging to a process $\mc Y$. The key difference with the previous section is that now the term (\ref{eq:termSBE}) will be able to contribute.

\medskip

From Dynkin's formula, for $f\in\mathcal D(\bb T)$ we have that:
\begin{equation}\label{eq:Dynkin_xi}
\mc M_t^n(f)=  \mathcal{Y}^{n}_t  (f) -  \mathcal{Y}^{n}_0  (f)-\int_0^t{(\partial_s+\theta(n)\mc L  ) \mathcal{Y}^{n}_s  (f) }ds  
\end{equation}
is a martingale, where 
\begin{align}
(\partial_s+\theta(n)\mc L  )  \mathcal{Y}^{n}_s  (f) 
&=\frac{\gamma\theta(n)}{n^2}\mc Y^n_s(\Delta_n f)\notag\\& \quad - b^2\frac{\theta(n)}{n^{3/2}}\alpha_n\sum_{x \in \bb T_n} \big(\nabla_n T^+_{c_n s}f\big) \, (\tfrac{x}{n})\, \xi_x(s\theta(n))\xi_{x+1}(s\theta(n))\notag\\
&\quad +2b^2\rho \frac{\theta(n)}{n^{3/2}} \alpha_n\sum_{x\in\bb T_n}  \nabla\big(T^+_{c_n s}f \big)(\tfrac{x}{n})\, \xi_x(s\theta(n)).\label{eq:mart_dec_xi}
\end{align}
Now we can sum and subtract terms, perform a summation by parts and a Taylor expansion on $T^+_{c_n s}f$ to write the time integral of the r.h.s.~of \eqref{eq:mart_dec_xi} as 
\begin{align*}
&\frac{\gamma\theta(n)}{n^2}\int_0^t \mc Y^n_s(\Delta_n f)\, ds\\
&-b^2\frac{\theta(n)}{n^{3/2}}\alpha_n \int_0^t \sum_{x\in\bb T_n} \big(\nabla_n T^+_{c_n s} \, f \big) \big(\tfrac{x}{n}\big)\Big[\xi_x(s\theta(n))\xi_{x+1}(s\theta(n))-\rho\xi_x(s\theta(n)) \Big.\\
&\hspace{8cm} \Big.-\rho\xi_{x+1}(s\theta(n))\Big] \, ds,
\end{align*}
plus a term whose variance is $\mathcal{O}(\theta(n)^2\;\alpha_n^2 \; n^{-4})$.
Adding the constant $\rho^2$ above, which we can do since the sum of the discrete gradients vanishes on the periodic lattice: $\sum_{x\in\bb T_n} (\nabla_n T^+_{c_n s}f ) (\tfrac{x}{n})=0$,  we rewrite the last expression as 
\begin{align*}
&\frac{\gamma\theta(n)}{n^2} \int_0^t \mc Y^n_s(\Delta_n f)\, ds\\
& -b^2\frac{\theta(n)}{n^{3/2}}\alpha_n \int_0^t \sum_{x\in\bb T_n} \big( \nabla_n T^+_{c_n s}f \big)\big(\tfrac xn\big)\,  \bar{\xi}_x(s\theta(n))\bar{\xi}_{x+1}(s\theta(n))\, ds.
\end{align*}
Then the martingale decomposition for the field $\mc Y^n_t$ defined in \eqref{eq:xi_field} is given by
\begin{align}
\label{eq:mart_dec_xi_field_2_longer1}
\mc M_t^n(f)& =  \mathcal{Y}^{n}_t  (f) - 
 \mathcal{Y}^{n}_0  (f)-\int_0^t  \frac{\gamma\theta(n)}{n^2}\mc Y^n_s(\Delta_n f)\, ds\\
&\quad + \int_0^tb^2\frac{\theta(n)}{n^{3/2}}\alpha_n\sum_{x\in\bb T_n} \big(\nabla_n T^+_{c_n s}f \big) \big(\tfrac xn\big)  \, \bar{\xi}_x(s\theta(n))\bar{\xi}_{x+1}(s\theta(n))\, ds. \label{eq:mart_dec_xi_field_2_longer}
\end{align}
Observe that by using the bound (\ref{eq:CSs}), the last term  \eqref{eq:mart_dec_xi_field_2_longer} has variance of order at most $\alpha_n^2\; (\theta(n))^{2} \; n^{-2}$. This bound is not sharp and can be improved by a $H_{-1}$ estimate. From Theorem 4 of \cite{BG} the term \eqref{eq:mart_dec_xi_field_2_longer} has variance of order at most $\alpha_n^2 \; (\theta(n))^{3/2} \; n^{-2}$. Indeed, by looking into the proof of Theorem 4 in \cite{BG}, for a  function $\psi:{\bb R_+}\times \bb T\rightarrow{\bb R}$, it holds that if $t\le T$, 
\begin{equation}\label{eq:BG_bound}
\bb E\Big[ \Big(\int_0^t\sum_{x\in\bb T_n} \psi  (s,\tfrac xn)  \, \bar{\xi}_x(s\theta(n))\bar{\xi}_{x+1}(s\theta(n))\, ds\Big)^2\Big]\lesssim \frac{n}{\sqrt {\theta(n)}}\int_0^t \|\psi(s, \cdot) \|_{2,n}^2\, ds
\end{equation} 
where 
\begin{equation}\label{vsummable}
 \|\psi(s, \cdot) \|_{2,n}^2:=\frac{1}{n}\sum_{x\in\mathbb{T}_n}\psi^2 \big(s, \tfrac{x}{n}\big).
\end{equation}
From this we easily get the last bound. 
Therefore if $a<\frac{4}{3}(\kappa+1)$ the  $\bb L^2$-norm of the  last term \eqref{eq:mart_dec_xi_field_2_longer} vanishes as $n\to+\infty$. Moreover, if $a<2$ the $\bb L^2$-norm of the integral term at the r.h.s.~of \eqref{eq:mart_dec_xi_field_2_longer1} vanishes as $n\to+\infty$ (by the rough bound provided by \eqref{eq:CSs}). 

\medskip

Let us now study the martingale $\mathcal{M}^n$ appearing in \eqref{eq:mart_dec_xi_field_2_longer1}--\eqref{eq:mart_dec_xi_field_2_longer}. By Dynkin's formula, the quadratic variation of the martingale is given by
\begin{equation}\label{eq:qv_xi}
\langle\mc M^n(f)\rangle_t=  \int_0^t \left\{ \theta(n)\mc L(\mc Y_s^n(f))^2-2\mc Y_s^n(f)\theta(n)\mc L \mc Y_s^n(f)\right\}\,  ds, 
\end{equation}
and a simple computation shows that last display is equal to
\begin{equation}\label{eq:qv_mart_xi_field}
\langle\mc M^n(f)\rangle_t= \gamma  \int_0^t \frac{\theta(n)}{n}\sum_{x\in\bb T}\big(f(\tfrac {x+1}{n})-f(\tfrac{x}{n})\big)^2\big(\xi_{x+1}(\theta(n)s)-\xi_x(\theta(n)s)\big)^2ds.
\end{equation}
Recall that  $\theta(n)=n^a$ and $\alpha_n=\alpha n^{-\kappa}$. 
\begin{itemize}
\item If $a=2$ and $\kappa>\frac{3}{4}a-1$ then 
\begin{equation}\label{eq:mart_dec_xi_field_2_a=2}
\mc M_t^n(f)=  \mathcal{Y}^{n}_t  (f) -  \mathcal{Y}^{n}_0  (f)-\gamma \int_0^t  \mc Y^n_s(\Delta_n f)ds 
\end{equation}
plus a term that vanishes in $\bb L^2$ as $n\to+\infty$. 
Moreover, the quadratic variation of the martingale satisfies:
\begin{equation}\label{eq:qv_mart_xi_fieldlim}
\lim_{n\to+\infty}\bb E[\langle\mc M^n(f)\rangle_t]= 2t  \gamma \tau \|\nabla f\|_0^2.
\end{equation}
Then $(\mc Y^n)_n$ converges to the solution of the Ornstein Uhlenbeck equation: \begin{equation}\label{eq:ou}
d\mc Y_t=\gamma\Delta \mc Y_t dt+\sqrt {2\gamma\tau}\nabla d\mc B_t.
\end{equation}

\medskip
\item If $a<2$ and $\kappa>\frac{3}{4}a-1$ then 
\begin{equation}
\mc M_t^n(f)=  \mathcal{Y}^{n}_t  (f) - \mathcal{Y}^{n}_0  (f)
\end{equation}
plus a term that vanishes in $\bb L^2$ as $n\to+\infty$. Moreover, the quadratic variation of the martingale satisfies:
\begin{equation*}
\lim_{n\to+\infty}\bb E[\langle\mc M^n(f)\rangle_t]= 0.
\end{equation*}
 Then $\mc Y$ has a trivial evolution given by: 
\begin{equation}
\label{eq:trivial}
\mc Y_t(f)=\mc Y_0(f), \quad\textrm{so that } \quad d\mc Y_t=0.
\end{equation}
\medskip

\item If $a=2$ and $\kappa=\frac{3}{4}a-1=\frac12$ then 
\begin{equation}\begin{split}\label{eq:mart_dec_xi_SBE}
\mc M_t^n(f)=  \mathcal{Y}^{n}_t  (f) - & \mathcal{Y}^{n}_0  (f)-\gamma \int_0^t \mc Y^n_s(\Delta_n f)ds\\
+&b^2\alpha \int_0^t\sum_{x\in\bb T_n} \big(\nabla_n T^+_{c_n s}\, f \big)\big(\tfrac xn\big)  \bar{\xi}_x(sn^2)\bar{\xi}_{x+1}(sn^2)ds. \end{split}
\end{equation}
plus a term which vanishes in $\bb L^2$ as $n\to+\infty$. 
Now we recall from  \cite{GJSim} a second-order Boltzmann-Gibbs principle which is needed in order to close the last term at the r.h.s.~of last expression in terms of the fluctuation field $\mc Y^n$. 

\begin{theorem}[Second-order Boltzmann-Gibbs principle]\label{theo:BG}
Fix a function $\psi:\bb R_+\times\mathbb{T}_n\to{\mathbb{R}}$ such that $$\int_0^t \|\psi(s, \cdot) \|_{2,n}^2\, ds<\infty.$$
For any $t \in [0,T]$, any positive integer $n$ and any $\varepsilon \in (0,1)$, it holds that:
\begin{multline}
\label{eq:BGexpo}
\mathbb{E}\Big[\Big(\int_{0}^t \sum_{x\in\mathbb{T}_n} \psi(s,\tfrac xn)\Big\{\bar{\xi}_{x}(sn^{2})\bar{\xi}_{x+1}(sn^{2})-\big(\vec{\xi}_{x}^{\varepsilon n}(sn^{2})\big)^2 +\frac{\tau^2}{\varepsilon n}\Big\}\; ds \Big)^2\Big]\\
\lesssim \int_0^t \|\psi(s, \cdot) \|_{2,n}^2\, ds\Big\{\varepsilon+\frac{t}{\varepsilon^2n}\Big\},
\end{multline}
where $\vec{\xi}^{\varepsilon n}_x$ is the empirical average on the box of size $\lfloor \varepsilon n \rfloor$  at the right of site $x$:
\begin{equation}
\vec{\xi}^{\varepsilon n}_x=\frac{1}{\lfloor\varepsilon n \rfloor}\sum_{y=x+1}^{x+\lfloor\varepsilon n\rfloor}\bar{\xi}_y. \label{eq:mean}
\end{equation}
\end{theorem}

\begin{proof}
The proof of last result is completely analogous to the proof of Theorem  1 in \cite{GJSim} and for that reason it is omitted.
\end{proof}

From the previous theorem,  we can replace (in $\bb L^2$) the last term at the r.h.s.~of \eqref{eq:mart_dec_xi_SBE}, for $n$ sufficiently big  and then $\varepsilon$ sufficiently small, by 
\begin{equation}
b^2\alpha\int_0^t\sum_{x\in\bb T_n} \big( \nabla_n T_{c_n s }^+f \big) \big(\tfrac xn \big)\big(\bar{\xi}_x^{\eps n}(sn^2)\big)^2 ds.
\end{equation}
Note that last expression writes as 
\begin{equation}\label{eq:kpz_1}
b^2\alpha\int_0^t\sum_{x\in\bb T_n}\big(\nabla_n T_{c_n s }^+\, f\big)\big(\tfrac xn\big)\Big(\frac {1} {\eps n}\sum_{y=x}^{x+\eps n}\bar{\xi_x}(sn^2)\Big)^2 ds.
\end{equation}
For $\varepsilon>0$ and $x\in\bb T$, we recall that $\iota_\varepsilon(x):\bb T\to\bb R$ is the function defined for $y\in \bb T$ by  $\iota_\varepsilon(x)(y)=\varepsilon^{-1} \, \textbf{1}_{x<y\leq x+\varepsilon}$. Note that 
\begin{equation}\label{eq:kpz_2}
\frac{1}{\sqrt n}\mc Y^n_{s}(\iota_\eps(x))=\frac{1}{\eps n}\sum_{y=x-2b^2\rho  \alpha n^{3/2} s}^{x+\eps n-2b^2\rho \alpha n^{3/2} s}\bar{\xi_y}(sn^2).
\end{equation}
If in  \eqref{eq:kpz_1} we change the variable $x$ into  $z-2b^2\rho\alpha   n^{3/2} s$, we rewrite \eqref{eq:kpz_1} as
\begin{equation}
b^2\alpha\int_0^t\sum_{z\in\bb T_n} \nabla_n f(\tfrac {z}{n})\Big(\frac {1} {\eps n}\sum_{y=z-2b^2\rho  \alpha n^{3/2} s}^{z-2b^2\rho\alpha   n^{3/2} s+\eps n}\bar{\xi_y}(sn^2)\Big)^2 ds
\end{equation}
and from \eqref{eq:kpz_2} last expression writes as 
\begin{equation}
b^2\alpha\int_0^t\frac 1n\sum_{z\in\bb T_n} \nabla_n f(\tfrac {z}{n})\big(\mc Y^n_{s}(\iota_\eps(z))\big)^2ds.
\end{equation}
Then we get
\begin{align*}
\mc M_t^n(f)=  \mathcal{Y}^{n}_t  (f) &-  \mathcal{Y}^{n}_0  (f)-\gamma\int_0^t  \mathcal{Y}^{n}_s  (\Delta_n f) ds \\ &+b^2\alpha\int_0^t\frac 1n\sum_{z\in \bb T_n} \nabla_n f(\tfrac {z}{n})\big(\mc Y^n_{s}(\iota_\eps(z))\big)^2ds.\end{align*}
Moreover, the quadratic variation of the martingale satisfies:
\begin{equation*}
\lim_{n\to+\infty}\bb E[\langle\mc M^n(f)\rangle_t]= 2t  \gamma \tau \|\nabla f\|_2^2
\end{equation*}
so that $\mc Y$ is solution of the stochastic Burgers equation:
\begin{equation}
d\mc Y_t=\gamma\Delta \mc Y_t dt+b^2\alpha \nabla (\mc Y_t)^2 dt + \sqrt{2\gamma\tau}\nabla d\mc B_t.
\end{equation}

\end{itemize}

\section{The limit of the sequence of martingales  $(\mathcal N^n )_{n\in \mathbb N}$} \label{sec:limit}

In this section we prove convergence of the sequence of martingales $$\left\{\mathcal N_t^n \bigcdot \vec f\; ; \; t\in [0,T]\right\}_{n\in \mathbb N}$$ which is a consequence of Theorem VIII.3.12  in \cite{J.S.} which can be stated as follows.

\begin{proposition}[\cite{J.S.}]
Let $t \in [0,T] \mapsto C_t  \in [0, \infty)$ be a deterministic continuous  function of the time $t$. Let $\{M_t^n\; ;\;  t \in [0,T]\}_{n \in \mathbb N}$ be a sequence of square-integrable real-valued martingales with c\`adl\`ag trajectories defined on a probability space $(\Omega,\mc F, \bb P)$. Let $\{\langle M^n\rangle_t\; ;\;  t \in [0,T]\}$ denote the quadratic variation of $\{M_t^n\; ; \; t \in [0,T]\}$. Assume that
\begin{itemize}
\item[i)] For each $n\in \bb N$, the quadratic variation process $\{\langle M^n\rangle_t\; ;\;  t \in [0,T]\}$ has continuous trajectories $\PP$ a.s.;
\item[ii)] the maximal jump satisfies
\begin{equation}\label{eq:maximal_jump}
\lim_{n \to \infty} \bb E\Big[ \sup_{0 \le s \le T} \big|M_s^n-M_{s-}^n\big| \Big] =0;
\end{equation}
Above $\bb E$ denotes the expectation w.r.t. $\bb P$.
\item[iii)] For each $t \in [0,T]$, the sequence of random variables $\{\langle M^n\rangle_t \}_{n \in \bb N}$ converges in probability to the deterministic path  $\{C_t \; ;\;  t \in [0,T]\}$.
\end{itemize}
Then the sequence $\{M_t^n\; ;\;  t \in [0,T]\}_{n \in \bb N}$ converges in law in $\bb D ([0,T], \RR)$ to a martingale $\{M_t\; ;\; t \in [0,T]\}$ with quadratic variation $t \mapsto C_t$. Moreover $\{M_t\; ;\; t \in [0,T]\}$ is a mean zero Gaussian process. 
\end{proposition}
\begin{remark}
We note that if in the previous theorem $C_t=\sigma^2t$, then the limit is a Brownian motion with quadratic variation  equal to $\sigma^2t$.
\end{remark}

Before we proceed we explain how to deduce  Proposition \ref{martconv} from Theorem VIII.3.12 of \cite{J.S.}. To get Proposition \ref{martconv}, we use the statement of Theorem VIII.3.12 which requires assumptions (3.14)  and  b) (iv) (both in \cite{J.S.}) to get the convergence in law of the martingales. By the assertion VIII.3.5 in \cite{J.S.}, the  conditions
[$\hat\delta_5-D$] and (3.14) are a consequence of \eqref{eq:maximal_jump} above. Moreover, condition $[\gamma_5-D]$ defined in (3.3) page 470 of \cite{J.S.}   is a consequence of iii) above.

As a consequence of last result we conclude that:
\begin{corollary}
\label{martconv}
Let $\vec f = (f^1,f^2)^\dagger \in\mc D(\bb T, \bb R^2)$. The sequence of martingales $\{\mc N^n_t \bigcdot \vec f :t\in [0,T]\}_{n\in \bb N}$ converges in  law under the topology of $\bb D([0,T], \bb R)$, as $n\to\infty$, to: 
\begin{itemize} \item  $0$ when $a<2$ and $\kappa\leq{a-1}$ ; 
\medskip

\item to a martingale $\{\mathcal N_t \bigcdot \vec f:t\in [0,T]\}$ which is a mean-zero Gaussian process  and whose quadratic variation is given by
\medskip

\begin{itemize} \item[$\circ$]  if $\kappa>1$ and $a=2$:
\begin{equation*}
\langle\mc N \bigcdot \vec f \, \rangle_t= 2t \gamma \tau^2\|\nabla f^1\|_0^2+2t \gamma \sigma^2\|\nabla f^2\|_0^2 + 4\gamma \delta t \langle \nabla f^1, \nabla f^2 \rangle_0 \, ;
\end{equation*} 

\item[$\circ$] if $\kappa=1$ and $a=2$:
\begin{equation*}
\langle\mc N \bigcdot \vec f \, \rangle_t= 2t \gamma \tau^2\|\nabla f^1\|_0^2+2t \gamma \sigma^2\|\nabla f^2\|_0^2 + 4\gamma \delta c^{-1} \, \langle T^+_{ct} f^1 -f^1 \, ,\,  \nabla f^2\rangle_0.
\end{equation*} 
\end{itemize}\end{itemize}
\end{corollary}

\begin{proof}
Now we fix $\vec f=(f^1,f^2)^\dagger \in\mc D(\bb T,\RR^2)$. In order to apply Proposition \ref{martconv} to the sequence $(\mathcal N^n \bigcdot \vec f \,)_{n\in \mathbb N}$, we note that item i) is trivial because of \eqref{eq:qv_mart_chi_field_2}. Note that ii) is a consequence of the computations performed  in the proof of Lemma \ref{lem:max_jumps}. 
Finally we prove iii),
that is,  the convergence, in the $\bb L^2$-norm, of the quadratic variation of $\mc N^n \bigcdot \vec f$. For that purpose, note that, by using the inequality $(x+y)^2\leq 2x^2+2y^2$ we can bound from above, 
\begin{equation*}
\bb  E\Big[\Big(\langle\mc N^n \bigcdot \vec f \rangle_t- \bb E\big[\langle\mc N^n \bigcdot \vec f\rangle_t\big]\Big)^2\Big]
\end{equation*}
by a constant times the sum of the next three terms:
\begin{align*}
\mathrm{I}&:=\bb  E\Big[\Big( \int_0^t \frac{\theta(n)}{n^3}\sum_{x\in\bb Z}\big(\nabla_nT^+_{c_ns}f^1(\tfrac{x}{n})\big)^2\big\{\big(\xi_{x+1}(\theta(n)s)-\xi_x(\theta(n)s)\big)^2-2\tau^2\big\} ds\Big)^2\Big] \\
\mathrm{II}&:=\bb  E\Big[\Big( \int_0^t \frac{\theta(n)}{n^3}\sum_{x\in\bb Z}\big(\nabla_n f^2 (\tfrac{x}{n})\big)^2\big\{\big(\eta_{x+1}(\theta(n)s)-\eta_x(\theta(n)s)\big)^2-2\sigma^2\big\} ds\Big)^2\Big]\\
\mathrm{III}&:= \bb  E\Big[\Big( \int_0^t \frac{\theta(n)}{n^3}\sum_{x\in\bb Z}\big(\nabla_nT^+_{c_ns}f^1\big) (\tfrac{x}{n})(\nabla_ng)(\tfrac{x}{n})\big\{\big(\eta_{x+1}(\theta(n)s)-\eta_x(\theta(n)s)\big)\\
& \qquad \quad \times \big(\xi_{x+1}(\theta(n)s)-\xi_x(\theta(n)s)\big)-2(\delta-\rho v)ds\big\}\Big)^2\Big]. \end{align*}
A simple computation, based on (\ref{eq:CSs}), shows that each one of the  last three expectations are of order $\mc O((\theta(n))^2\; n^{-5})$. Since $\theta(n)\leq n^2$, the last three  terms vanish as $n\to+\infty$. From these computations we conclude that when $a<2$ we have that $C_t=0$, so that the sequence $(\mathcal N_t^n \bigcdot \vec f \, )_{n\in \mathbb N}$
converges to $0$ as $n\to+\infty$, but when $a=2$ and $k>1$, $$C_t=2t \gamma \tau^2\|\nabla f^1\|_0^2 +2t \gamma \sigma^2\|\nabla f^2\|_0^2+4\gamma\delta \langle \nabla f^1,\nabla f^2\rangle_0$$
and when $a=2$ and $k=1$ we have that  
$$C_t=2t \gamma \tau^2\|\nabla f^1\|_0^2 +2t \gamma \sigma^2\|\nabla f^2\|_0^2+4\gamma\delta c^{-1} \langle T^+_{ct}f^1-f^1,\nabla f^2\rangle_0.$$
 From Proposition  \ref{martconv} we conclude that the sequence $\{\mathcal N_t^n\,; \,t \in [0,T]\}_{n \in \bb N}$  converges to a mean-zero Gaussian process $\{\mathcal N_t\,;\, t \in [0,T]\}$ which is a martingale with  quadratic variation given by $C_t$. This finishes the proof of Corollary \ref{martconv}.
\end{proof}

\section{Tightness}  \label{sec:tight}

\subsection{Tightness for the $\mc Z$ field}
\label {sec:tightness_chi}

In this section we prove tightness of the sequence $\{\mathcal X_t^n= \mc Z_t^n \bigcdot \; ;\;  t \in [0,T]\}_{n \in \bb N} \in \bb D ([0,T], (\bb H_{k}\times \bb H_k )^\prime)$ following Chapter 11 of \cite{K.L.}. By Remark \ref{rem:topo_ident} this implies the tightness of $\{\mc Z_t^n\; ;\;  t \in [0,T]\}_{n \in \bb N} \in \bb D ([0,T], \bb H_{-k}\times \bb H_{-k})\}_{n \in \bb N}$. We assume that $a \le \inf(2, \kappa+1)$. We need to show that:
\begin{enumerate}[\hspace{1cm} (1)\quad ]
\item[(A)] $ \lim_{A\rightarrow{+\infty}}\limsup_{n\rightarrow{+\infty}}\mathbb{P} {\Big(\sup_{0\leq{t}\leq{T}}\|\mc Z_{t}^n \bigcdot \|_{-k}^{2}>A\Big)}=0,$
\medskip

\item [(B)] $\forall{\epsilon>{0}},\quad \lim_{\delta\rightarrow{0}}\limsup_{n\rightarrow{+\infty}}\mathbb{P} \Big(\omega_{\delta}(\mc Z^n \bigcdot)\geq{\epsilon}\Big)=0,$
\end{enumerate}
where
\begin{equation*}
\omega_{\delta}(\mc Z^n \bigcdot):=\sup_{\substack{|s-t|<\delta\\0\leq{s,t}\leq{T}}}\|\mc Z^n_{t} \bigcdot -\mc Z^n_{s} \bigcdot\|_{-k}.
\end{equation*}
The norms above and the corresponding inner products $\langle \cdot, \cdot\rangle_{-k}$ have been introduced in the beginning of Section \ref{sec:stat_res}. We start by showing condition (A). For each integer $z\in \bb Z$, recall that $h_{z}$ denotes the function defined in \eqref{eq:defh}.

\begin{lemma} 
\label{th:lemmatight}
Assume that $a \le \inf(2, \kappa+1)$. There exists a finite constant $C(T)>0$ such that for every $z\in\bb Z$,
\begin{equation*}
\sup_{n\rightarrow{+\infty}}\mathbb{E}\Big[\sup_{0\leq{t}\leq{T}} \,\big|\mc Z^n_{t} \bigcdot (h_z,0)^\dagger\big|^2+\big|\mc Z_t^n \bigcdot (0,h_z)^\dagger\big|^{2}\Big] \; \leq C(T) \; (1+ z^4\mathbf{1}_{a=2}+z^2).
\end{equation*} 
\end{lemma}

\begin{proof}Recall \eqref{eq:Dynkin_eta}:
\begin{equation} 
\mc N_t^n \bigcdot \vec f=  \mathcal{Z}^{n}_t   \bigcdot \vec f -  \mathcal{Z}^{n}_0   \bigcdot \vec f -\int_0^t(\partial_s+\theta(n)\mc L  ) \mathcal{Z}^{n}_s   \bigcdot \vec f \, ds.
\end{equation}
A simple computation shows that for $ \vec f \in\{(h_z,0)^\dagger, (0,h_z)^\dagger\}$
\begin{equation*}
\lim_{n\rightarrow{+\infty}}\mathbb{E}\big[(\mathcal{Z}^{n}_0 \bigcdot \vec f )^{2}\big]\lesssim \|h_{z}\|_0^2=1.
\end{equation*}
To treat the martingale term, we rely on Doob's inequality to get that  for $\vec f\in\{(h_z,0)^\dagger, (0,h_z)^\dagger\}$
\begin{equation*}
 \mathbb{E}\Big[\sup_{0\leq{t}\leq{T}}\big|\mc N_t^n \bigcdot \vec f\big|^{2}\Big]\leq{4
\mathbb{E}\Big[\big|\mc N_T^n \bigcdot \vec f \big|^{2}\Big]}.
\end{equation*}
From \eqref{eq:qv_mart_chi_field_2} and since $\vec f \in\{(h_z,0)^\dagger, (0,h_z)^\dagger\}$, it follows that
\begin{equation}
\limsup_{n\rightarrow{+\infty}} \mathbb{E}\Big[\sup_{0\leq{t}\leq{T}}\big|\mc N_t^n \bigcdot \vec f\big|^{2}\Big]\; \lesssim \; T\frac{\theta(n)}{n^2}\|\nabla h_z\|_0^2\;  \lesssim\;  T z^2
\end{equation}
since $\theta(n)=n^a$ and $a \le 2$ and $\|\nabla h_z\|_0^2 \lesssim z^2$. Finally, it remains to bound:
\begin{equation*}
\mathbb{E}\Big[\sup_{0\leq{t}\leq{T}}\Big(\int_0^t(\partial_s+\theta(n)\mc L  ) \mathcal{Z}^{n}_s \bigcdot \vec f \, ds\Big)^{2}\Big]
\end{equation*}
for $\vec f \in\{(h_z,0)^\dagger, (0,h_z)^\dagger\}$.
From the computations following \eqref{eq:Dynkin_eta}, we see that last expectation is bounded from above by a constant times
\begin{equation*}
T^2\frac{\theta(n)^2}{n^4}\|\Delta h_z\|_0^2+T^2\frac{\theta(n)^2\alpha_n^2}{n^2}\|\nabla h_z\|_0^2 \lesssim n^{a-2} z^4 + n^{2(a-\kappa-1)}z^2 
\end{equation*}
The limit as $n\to+\infty$ of last term is equal to $z^4$ if  $a=2$ and equal to $z^2$  if $a=k+1<2$, otherwise it is $0$.
This proves  the lemma. 
\end{proof}

\begin{remark}
We observe here that in the regime $a>\kappa +1 $ (e.g. $a=2$, $\kappa=1/2$) the previous bound is not sufficiently sharp to prove tightness of the field $\mc Z^n$. We will need to show tightness of the $\mc Y^n$ field in the next section in this range of parameters and we will have to deal with this problem. \end{remark}

\begin{corollary}
\label{eq:corPat}
Assume that $a\le \inf(2, \kappa+1)$ and  $k>5/2$. It holds that
\begin{enumerate}[\rm (1)\quad ]
\item $\limsup_{n\rightarrow{+\infty}}\mathbb{E}\Big[\sup_{0\leq{t}\leq{T}}\|\mc Z^n_{t} \bigcdot \|_{-k}^{2}\Big]<{\infty}.$
\item  $\limsup_{j\rightarrow{+\infty}}\limsup_{n\rightarrow{+\infty}}\\ \displaystyle\mathbb{E}\Big[\sup_{0\leq{t}\leq{T}}\sum_{|z|\geq{j}}\Big\{\big|\mc Z^n_{t} \bigcdot (h_z,0)^\dagger \big|^2+ \big|\mc Z^n_{t} \bigcdot (0,h_z)^\dagger \big|^2\Big\}\gamma_{z}^{-k}\Big]=0.$
\end{enumerate}
\end{corollary}

\begin{proof}
The first item of the corollary follows from  \eqref{eq:chi_field}, the  definition of  the inner product $\langle \cdot , \cdot \rangle_{-k}$ of $\bb{H}_{-k}\times \bb H_{-k}$ given in \eqref{eq:inner_product_H_minus_k_times_minus_k}  and  the previous lemma.
More precisely:  
\begin{equation*}
\begin{split}
\bb{E}\Big[\sup_{0\leq{t}\leq{T}}\|\mc Z^n_{t} \bigcdot\|_{-k}^{2}\Big]&=\sum_{z\in{\mathbb{Z}}} \gamma_{z}^{-k}\bb{E}\left[\sup_{0\leq{t}\leq{T}}\big|\mc Z^n_{t} \bigcdot (h_z,0)^\dagger\big|^2+ \big|\mc Z^n_{t} \bigcdot (0,h_z)^\dagger \big|^2\right]\\
&\lesssim \sum_{z\in{\mathbb{Z}}} \frac{1}{(1+(2\pi z)^2)^{(k-2)}}
\end{split}
\end{equation*}
and last sum is finite as long as $k>5/2$. The second item follows exactly by the same argument.
\end{proof}

By Chebychev's inequality, condition (A) follows from (1) in Corollary \ref{eq:corPat} . It remains now to prove (B) but since (2) of Corollary \ref{eq:corPat}  holds, (B) follows from the next lemma.

\begin{lemma}
For every $j \ge 1$ and every $\epsilon>{0}$,
\begin{equation*}
\begin{split}
\lim_{\delta\rightarrow{0}}\limsup_{n\rightarrow{+\infty}} \; \mathbb{P}\Big[\sup_{\substack{|s-t|<\delta\\0\leq{s,t}\leq{T}}}\quad
\sum_{|z|\leq{j}}&\Big\{\big(\mc Z_{t}^n \bigcdot (h_z, 0)^\dagger -\mc Z_{s}^n \bigcdot (h_{z}, 0)^\dagger\big)^{2}\\&+\big(\mc Z_{t}^n \bigcdot ( 0,h_z)^\dagger -\mc Z_{s}^n \bigcdot (0, h_z)^\dagger\big)^{2}\Big\}\gamma_{z}^{-k}>\epsilon\Big]=0.
\end{split}
\end{equation*}
\end{lemma}
To prove last lemma it is enough to show, for every $z\in{\mathbb{Z}}$, $\epsilon>0$ and for $\vec f \in\{(h_z,0)^\dagger, (0,h_z)^\dagger\}$, that
\begin{equation*}
\lim_{\delta\rightarrow{0}}\limsup_{n\rightarrow{+\infty}}\;  \mathbb{P}\Bigg[\sup_{\substack{|s-t|<\delta\\0\leq{s,t}\leq{T}}}\quad
\big(\mc Z^n_t  \bigcdot \vec f -\mc Z^n_{s} \bigcdot \vec f \big)^{2}>\epsilon\Bigg]=0.
\end{equation*}
This is a consequence of the next two lemmas.

\begin{lemma}
\label{lem:max_jumps}
Let $\vec f \in\mc D(\bb T, \RR^2)$. For every $\epsilon>0$
\begin{equation*}
\lim_{\delta\rightarrow{0}}\limsup_{n\rightarrow{+\infty}} \; \mathbb{P} \Bigg[\sup_{\substack{|s-t|<\delta\\0\leq{s,t}\leq{T}}}\quad
\big|\mc N_{t}^{n}  \bigcdot \vec f -\mc N_{s}^{n} \bigcdot \vec f \big|>\epsilon\Bigg]=0.
\end{equation*}
\end{lemma}

\begin{proof}
Denote by $\omega'_{\delta}(\mc N^n \bigcdot \vec f)$ the modified modulus of continuity defined by
\begin{equation*}
\omega'_{\delta}(\mc N^n  \bigcdot \vec f )=\inf_{\substack{\{t_{i}\}}}\quad\max_{\substack{0\leq{i}\leq{r}}}\quad\sup_{\substack{t_{i}\leq{s}<{t}\leq{t_{i+1}}}}\big|\mc N_{t}^{n}
 \bigcdot \vec f -\mc N_{s}^{n} \bigcdot \vec f \big|,
\end{equation*}
where the infimum is taken over all partitions of $[0,T]$ such that $$0=t_{0}<t_{1}<...<t_{r}=T$$ with $t_{i+1}-t_{i}>\delta$ for
$0\leq{i}\leq{r}$.
Since
\begin{equation*}
\omega_{\delta}(\mc N^{n} \bigcdot \vec f)\leq{2\omega'_{\delta}(\mc N^n \bigcdot \vec f)+\sup_{\substack{t}}\big|\mc N_{t}^{n} \bigcdot \vec f-\mc N_{t {-}}^{n} \bigcdot \vec f\big|}
\end{equation*}
it is sufficient to control the two terms on the r.h.s.~of the previous inequality separately. We start with the first one. Observe that
\begin{equation*}
\sup_{\substack{t}}\Big|\mc N_{t}^{n} \bigcdot \vec f -\mc N_{t {-}}^{n} \bigcdot \vec f \Big|=\sup_{\substack{t}}\Big|
\mc Z_t^n \bigcdot \vec f -\mc Z_{t-}^{n} \bigcdot \vec f \Big|
\end{equation*}
and that  for any $t\in[0,T]$ it holds that
\begin{equation}
\label{eq:eqcoco}
\Big|\mc Z_{t}^{n} \bigcdot \vec f -\mc Z_{t {-}}^{n} \bigcdot \vec f\Big|\leq \big|\mc Y_t^n(f^1)-\mc Y_{t-}^n(f^1)\big|+ \big|\mc V_t^n(f^2)-\mc V_{t-}^n(f^2)\big|.
\end{equation}
Let $x_t$ be the site where the jump occurred at time $t$ for the speeded by $\theta (n)$ process. Now we calculate the $\bb L^2$-norm of each term on the RHS of \eqref{eq:eqcoco} separately. Note that
\begin{equation*}
\begin{split}
&\bb  E\left[ \sup_{t \leq T} \big(\mc Y_t^n(f^1)-\mc Y_{t-}^n(f^1)\big)^2 \right]\\
&=\bb  E\left[ \sup_{t \leq T} \Big(\tfrac{1}{\sqrt n}\Big(T^+_{c_nt}f^1(\tfrac {x_t}{n})-T^+_{c_nt}f^1(\tfrac {x_{t}+1}{n})\Big) \, \Big(\xi_{x_t}(t\theta(n))-\xi_{x_t+1}(t\theta(n))\Big)\Big)^2 \right]\\
&\lesssim \frac{1}{n^3}\bb  E\left[ \sup_{t \leq T} \Big(\xi_{x_t}(t\theta(n))-\xi_{x_t+1}(t\theta(n))\Big)^2 \right]\\
&\lesssim \frac{1}{n^3}\bb  E\left[ \sup_{t \leq T} \Big(\sum_{x\in\bb T_n}(\xi_{x}(t\theta(n))-\xi_{x+1}(t\theta(n)))\Big)^2 \right]\\
&\lesssim \frac{1}{n^3}\bb  E\left[ \sup_{t \leq T} \Big(\sum_{x\in\bb T_n}\xi_{x}(t\theta(n))\Big)^2 \right]\\
&=\frac{1}{n^3}\bb  E\Big[ \Big(\sum_{x\in\bb T_n}\xi_{x}(0)\Big)^2 \Big]
\end{split}
\end{equation*}
where the last equality follows from the conservation of  $\sum_{x\in\bb T_n}\xi_x$.  A simple computation shows that the last term above is of order $\mc O(n^{-1})$. 
Finally,
\begin{equation*}
\begin{split}
&\bb  E\big[ \sup_{t \leq T}\big(\mc V_t^n(f^2)-\mc V_{t-}^n(f^2)\big)^2 \big]\\&=\bb  E\Big[ \sup_{t \leq T} \Big(\tfrac{1}{\sqrt n}\Big(f^2(\tfrac {x_t}{n})-f^2(\tfrac {x_{t}+1}{n})\Big)\Big(\eta_{x_t}(t\theta(n))-\eta_{x_t+1}(t\theta(n))\Big)\Big)^2 \Big]\\&\lesssim \frac{1}{n^3}\bb  E\Big[ \sup_{t \leq T} \Big(|\eta_{x_t}(t\theta(n))|+|\eta_{x_t+1}(s\theta(n))|\Big)^2 \Big]\\
&\lesssim \frac{1}{n^3}\bb  E\Big[ \sup_{t \leq T} \Big(\sum_{x\in\bb T_n}|\eta_{x}(t\theta(n))|+|\eta_{x+1}(t\theta(n))|\Big)^2 \Big]
\\
&\lesssim \frac{1}{n^3}\bb  E\Big[ \sup_{t \leq T} \Big(\sum_{x\in\bb T_n}|\eta_{x}(t\theta(n))|\Big)^2 \Big].
\end{split}
\end{equation*}
Since for any $b>0$ and for any $u\in\mathbb R$, it holds that $|u|\lesssim 1+V_b(u)$, last expectation can be bounded from above by 
\begin{equation*}
 \frac{1}{n^3} \bb  E\Big[ \sup_{t \leq T} \Big(\sum_{x\in\bb T_n}V_b(\eta_{x}(t\theta(n))\Big)^2 \Big]
\end{equation*}
plus a term that vanishes as $n\to +\infty$.
From the   conservation of  $\sum_{x\in\bb T_n}V_b(\eta_x)$, last display is equal to \begin{equation*}
 \frac{1}{n^3} \bb  E\Big[  \Big(\sum_{x\in\bb T_n}V_b(\eta_{x}(0)\Big)^2 \Big]
\end{equation*}
which is of order  $\mc O(n^{-1})$.
In order to finish the proof it is enough to  show that
\begin{equation*}
\lim_{\delta\rightarrow{0}}\limsup_{n\rightarrow{+\infty}}\; \mathbb{P} \Big[\omega'_{\delta}(\mc N^{n}  \bigcdot \vec f)>\epsilon\Big]=0
\end{equation*}
for every $\epsilon>0$. By Aldous' criterium, see, for example, Proposition 4.1.6 of \cite{K.L.},  it is enough to show that:
\begin{equation*}
\lim_{\delta\rightarrow{0}}\limsup_{n\rightarrow{+\infty}}\sup_{\substack{\tau\in{\mathfrak
{T}_{\tau}}\\0\leq{\theta}\leq{\delta}}}\; \mathbb{P} \Big[|\mc N_{\tau+\theta}^{n} \bigcdot \vec f -\mc N_{\tau}^{n} \bigcdot \vec f|>\epsilon\Big]=0
\end{equation*}
for every $\epsilon>0$. Here $\mathfrak {T}_{\tau}$ denotes the family of all stopping times bounded by $T$ with respect to the canonical
filtration. By Chebychev's inequality, the Optional Stopping Theorem and \eqref{eq:qv_mart_chi_field_2}
 the result follows.
\end{proof}
\begin{lemma}
Let $\vec f \in\{(h_z,0)^\dagger, (0,h_z)^\dagger\}$. For every $\epsilon>0$
\begin{equation*}
\lim_{\delta\rightarrow{0}}\limsup_{n\rightarrow{+\infty}}
\; \mathbb{P}\left[\sup_{\substack{|s-t|<\delta\\0\leq{s,t}\leq{T}}} \Big|
\int_{s}^{t}(\partial_r+\theta(n)\mc L  ) \, \mathcal{Z}^{n}_r  \vec f \, dr\Big|>\epsilon \right]=0.
\end{equation*}
\end{lemma}

\begin{proof}
By using the explicit expression for $(\partial_r+\theta(n)\mc L  ) \mathcal{Z}^{n}_r \bigcdot \vec f$, Chebychev's inequality and some simple computations  we conclude the proof.
\end{proof}
\begin{remark}
We observe that in the proof above it was not necessary to consider $\vec f \in\{(h_z,0)^\dagger, (0,h_z)^\dagger \}$. The proof works out for any $\vec f \in\mc D(\bb T)$.
\end{remark}

\subsection{Tightness for the $\mc Y$ field}
\label {sec:tightness_xi}

We note that the proof of tightness for the sequence $\{\mathcal Y_t^n\; ;\;  t \in [0,T]\}_{n \in \bb N}$ is completely similar to the one presented in the previous subsection. Observe that the field $\mc Y$ can be recovered from the field $\mc X$ by simply taking the test function $g=0$.  The only  regime which requires some care is when $a\leq \inf( \tfrac{4}{3}(\kappa+1),2)$ and $\kappa\leq 1$, so that we need to look carefully at  \eqref{eq:mart_dec_chi3}. See the comments at the end of the proof of Lemma \ref{th:lemmatight}. 
Now we treat this term. For that purpose, fix $\epsilon>0$ and note that, from \eqref{eq:BG_bound} we have that
\begin{multline*}
\bb E_n\Big[\sup_{0\leq t\leq T} \Big(\int_0^t \frac{\theta(n)\alpha_n}{n^{3/2}}\sum_{x\in\bb T_n} \big(\nabla_n T^+_{c_n s} \, h_z\big) (\tfrac xn) \, \bar{\xi}_x(sn^2)\bar{\xi}_{x+1}(sn^2)ds\Big)^2\Big]\\
\leq \frac{\theta(n)^{3/2}\alpha_n^2}{n^{2}}\int_{0}^T  \|\nabla_n T^+_{c_n s} \, h_z\|_{2,n}^2ds.
\end{multline*}
For the range of the parameters that we are looking at, last expression is bounded from above by $C(T)\int_{0}^T  \|\nabla_n T^+_{c_n s} \, h_z\|_{2,n}^2ds.$
The rest of the proof of tightness follows exactly from the same computations as presented above for the $\mc X$ field.

\section*{Acknowledgements}
This work has been supported by the projects EDNHS ANR-14- CE25-0011, LSD ANR-15-CE40-0020-01 of the French National Research Agency (ANR). This project has received funding from the European Research Council (ERC) under  the European Union's Horizon 2020 research and innovative programme (grant agreement   No 715734). The work of M.S. was also supported  by the Labex CEMPI (ANR-11-LABX-0007-01). We thank G\"unter Sch\"utz for his interest in this work.

\bibliographystyle{plain}

\end{document}